\documentclass[draft,reqno,T1]{amsproc}
\usepackage{tikz}
\usetikzlibrary{arrows,automata,petri,shapes,positioning,circuits.ee.IEC}
\usepackage{color}
\usepackage{amssymb}
\usepackage{amsmath}
\usepackage{amsfonts}
\usepackage{geometry}
\usepackage{bbm}
\usepackage{mathrsfs}
\usepackage{pbox}
\usepackage[tight]{subfigure}
\usepackage{multirow}
\usepackage[makeroom]{cancel}
\usepackage{braket}

\usepackage[all]{xy}

\definecolor{myurlcolor}{rgb}{0,0,0.4}
\definecolor{mycitecolor}{rgb}{0,0.5,0}
\definecolor{myrefcolor}{rgb}{0.5,0,0}
\usepackage[pagebackref,draft=false]{hyperref}
\hypersetup{colorlinks,
linkcolor=myrefcolor,
citecolor=mycitecolor,
urlcolor=myurlcolor}

\newcommand{\beq}{\begin{equation}}
\newcommand{\eeq}{\end{equation}}
\newcommand{\Z}{\mathbb{Z}}
\newcommand{\C}{\mathbb{C}}

\newcommand{\R}{\mathbb{R}}

\newcommand{\Rplus}{\mathbb{R}_{\geq 0}}
\renewcommand{\H}{\mathcal{H}}

\newcommand{\B}{\mathcal{B}}
\newcommand{\K}{\mathcal{K}}
\renewcommand{\S}{\mathcal{S}}

\newcommand{\lra}{\longrightarrow}
\newcommand{\eps}{\varepsilon}

\newcommand{\rsa}{\rightsquigarrow}

\newcommand{\tr}{\mathrm{tr}}
\newcommand{\id}{\mathrm{id}}

\newcommand{\defin}{:=}

\newcommand{\src}[1]{{\mathrm{src}(#1)}}
\newcommand{\tar}[1]{{\mathrm{tar}(#1)}}
\newcommand{\pa}[1]{{\mathrm{pa}(#1)}}
\newcommand{\ch}[1]{{\mathrm{ch}(#1)}}
\newcommand{\pst}[1]{{\mathrm{pst}(#1)}}

\newcommand{\inc}[1]{{\mathrm{in}(#1)}}
\newcommand{\out}[1]{{\mathrm{out}(#1)}}

\newcommand{\Cp}{\mathtt{C}}
\newcommand{\Dp}{\mathtt{D}}
\newcommand{\Pp}{\mathtt{C_1}}
\newcommand{\Qp}{\mathtt{D_1}}
\newcommand{\sop}{\mathtt{StOp}}
\newcommand{\fsop}{\mathtt{FinStOp}}
\renewcommand{\sp}{\mathtt{StOp_1}}
\newcommand{\qop}{\mathtt{QOp}}
\newcommand{\FinStoch}{\mathtt{FinStoch}}
\newcommand{\ground}{\begin{tikzpicture}[thick,scale=.8,circuit ee IEC] \node[ground,rotate=90]{}; \end{tikzpicture}} 

\swapnumbers
\theoremstyle{plain}
\newtheorem{thm}{Theorem}[section]
\newtheorem{lem}[thm]{Lemma}
\newtheorem{prop}[thm]{Proposition}

\newtheorem{conj}[thm]{Conjecture}

\newtheorem{defn}[thm]{Definition}

\theoremstyle{definition}

\theoremstyle{remark}
\newtheorem{rem}[thm]{Remark}

\numberwithin{equation}{section}
\newenvironment{partialproof}{\paragraph{\textsc{Partial Proof.}}}{\hfill$\square$\bigskip}
\newenvironment{sketchproof}{\paragraph{\textsc{Sketch of Proof.}}}{\hfill$\square$\bigskip}

\renewcommand{\labelenumi}{(\alph{enumi})}
\renewcommand{\theenumi}{(\alph{enumi})}

\allowdisplaybreaks

\renewcommand{\emph}[1]{\textbf{#1}}
\newcommand{\emphalt}[1]{\textit{#1}}

\begin{document}
\sloppy

\setlength{\jot}{6pt}



\title[Beyond Bell's Theorem]{Beyond Bell's Theorem II: Scenarios with arbitrary causal structure\medskip}

\address{Tobias Fritz, Perimeter Institute for Theoretical Physics\\
Waterloo ON, Canada}
\email{tfritz@perimeterinstitute.ca}



\thanks{\textit{Acknowledgements.} Thanks to all the participants of the Benasque meeting on Causal Structure in Quantum Theory, and in particular to the organizers Llu\'is Masanes and Rob Spekkens. Matthew Pusey has kindly made many helpful suggestions and contributed to some of our proofs. Research at Perimeter Institute is supported by the Government of Canada through Industry Canada and by the Province of Ontario through the Ministry of Economic Development and Innovation. The author has been supported by the John Templeton Foundation.}

\begin{abstract}
It has recently been found that Bell scenarios are only a small subclass of interesting setups for studying the non-classical features of quantum theory within spacetime. We find that it is possible to talk about classical correlations, quantum correlations and other kinds of correlations on any directed acyclic graph, and this captures various extensions of Bell scenarios which have been considered in the literature. From a conceptual point of view, the main feature of our approach is its high level of unification: while the notions of source, choice of setting and measurement play all seemingly different roles in a Bell scenario, our formalism shows that they are all instances of the same concept of ``event''.

Our work can also be understood as a contribution to the subject of causal inference with latent variables. Among other things, we introduce hidden Bayesian networks as a generalization of hidden Markov models.
\end{abstract}

\maketitle

\tableofcontents

\section{Introduction}
\label{introduction}

Bell's discovery of ``quantum nonlocality''~\cite{Bell,Bellreview} has maybe been one of the most influential discoveries in theoretical physics over the past 50 years. Not only has it been a vexing challenge and a source of philosophical discussion for many physicists interested in the foundations of quantum theory, it has also led more recently to the development of new information processing protocols whose performance and security rely crucially on violations of Bell inequalities~\cite{Ekert,VVqkd,PAal,VVrand}.

Over the course of time, it has been realized that Bell's \emphalt{gedankenexperiment}, whose causal structure is illustrated in Figure~\ref{Bellscen}, is only one of many possible setups in which the non-classicality of quantum theory manifests itself in correlations between between different observations distributed in spacetime. Generalizations of Bell's original idea include Svetlichny's extension from a two-party scenario to a many-party scenario~\cite{Svet}, Popescu's ``hidden nonlocality''~\cite{Pop} which modifies the Bell scenario by including additional preselection measurements (Figure~\ref{Popscen}), and the ``bilocality'' scenario of Branciard, Gisin and Pironio~\cite{BGP} (Figure~\ref{bilocal}). Subsequently, we proposed a general definition of what constitutes a \emphalt{correlation scenario}, showed how this definition comprises two-party and many-party Bell scenarios, the bilocality scenario, and many new ones (like Figure~\ref{trianglefig}), about which we also derived some first results~\cite{Fri}. Most recently, Gallego, W\"urflinger, Chaves, Ac\'in and Navascu\'es have considered ``sequential correlation scenarios''~\cite{GWCAN} which extend Popescu's idea by adding an additional choice of measurement setting (Figure~\ref{seqnl}).

The purpose of this present paper is threefold:

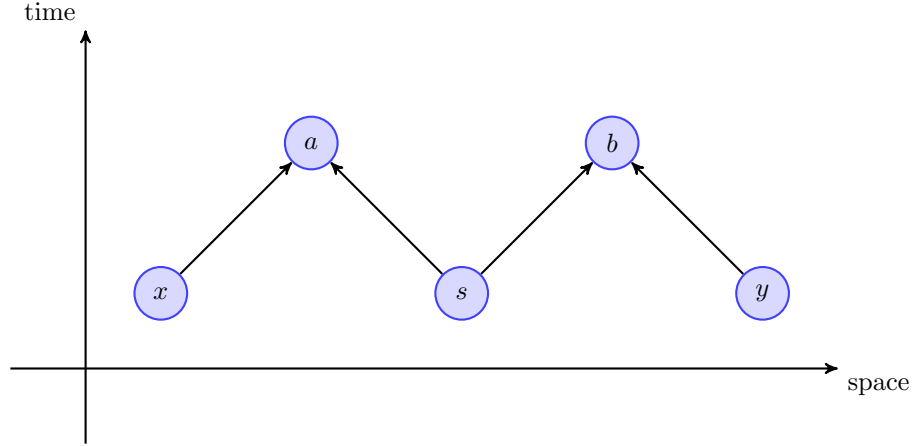
\begin{figure}
\centering
\begin{tikzpicture}[node distance=2.3cm,>=stealth',thick]
\draw[->] (-2,-1) -- (-2,4.5) node[anchor=south east] {time} ;
\draw[->] (-3,0) -- (8,0) node[anchor=north west] {space} ;
\tikzstyle{place}=[circle,thick,draw=blue!75,fill=blue!15,minimum size=7mm]
\node[place] (l) at (3,1) {$s$} ;
\node[place] (la) at (-1,1) {$x$} ;
\node[place] (lb) at (7,1) {$y$} ;
\node[place] (a2) at (1,3) {$a$} ;
\node[place] (b2) at (5,3) {$b$} ;
\draw[->] (la) -- (a2) ;
\draw[->] (lb) -- (b2) ;
\draw[->] (l) -- (a2) ;
\draw[->] (l) -- (b2) ;
\end{tikzpicture}
\caption{Causal structure of the two-party Bell scenario~\cite{Bell}.}
\label{Bellscen}
\end{figure}

\begin{enumerate}
\item We propose a general definition of a class of \emphalt{gedankenexperiments}, in which one can talk about classical correlations, quantum correlations and the like, and which comprises all the above scenarios as particular instances. We provide definitions of classical and quantum correlations and show that they correspond to the usual ones in the case of a Bell scenario. Moreover, we prove some very general technical results about classical correlations. 
\item We point out that our definition achieves a new level of unification: there is no distinction whatsoever between a ``source'', a ``choice of measurement setting'' and a ``measurement''. All three concepts are subsumed by our notion of \emph{event}, as illustrated in the Bell scenario of Figure~\ref{Bellscen}. A pair of events may or may not be connected by a causal link drawn as an arrow pointing from one event to the other. We think of an event as a measurement apparatus located at a specific point in spacetime which probes and potentially modifies the outside world and outputs a classical outcome. There are two ways to interpret the causal links: one can either think of them as representing a theoretical hypothesis about which events influence which other ones, or alternatively as physical systems propagating between the events---exactly as in Hardy's \emph{operator-tensor formulation} of quantum theory~\cite{HardyOT}. Both of these interpretations are equally viable, since the underlying mathematics is the same.
\begin{figure}
\centering
\begin{tikzpicture}[node distance=2.3cm,>=stealth',thick,scale=.9]
\draw[->] (-2.2,-1) -- (-2.2,6.5) node[anchor=south east] {time} ;
\draw[->] (-3.2,0) -- (12.3,0) node[anchor=north west] {space} ;
\tikzstyle{place}=[circle,thick,draw=blue!75,fill=blue!15,minimum size=7mm]
\node[place] (a1) at (3,3) {$a'$} ;
\node[place] (b1) at (7,3) {$b'$} ;
\node[place] (l) at (5,1) {$s$} ;
\draw[->] (l) -- (a1) ;
\draw[->] (l) -- (b1) ;
\node[place] (la) at (-1,3) {$x$} ;
\node[place] (lb) at (11,3) {$y$} ;
\node[place] (a2) at (1,5) {$a$} ;
\node[place] (b2) at (9,5) {$b$} ;
\draw[->] (la) -- (a2) ;
\draw[->] (lb) -- (b2) ;
\draw[->] (a1) -- (a2) ;
\draw[->] (b1) -- (b2) ;
\end{tikzpicture}
\caption{Causal structure of Popescu's scenario~\cite{Pop}.}
\label{Popscen}
\end{figure}
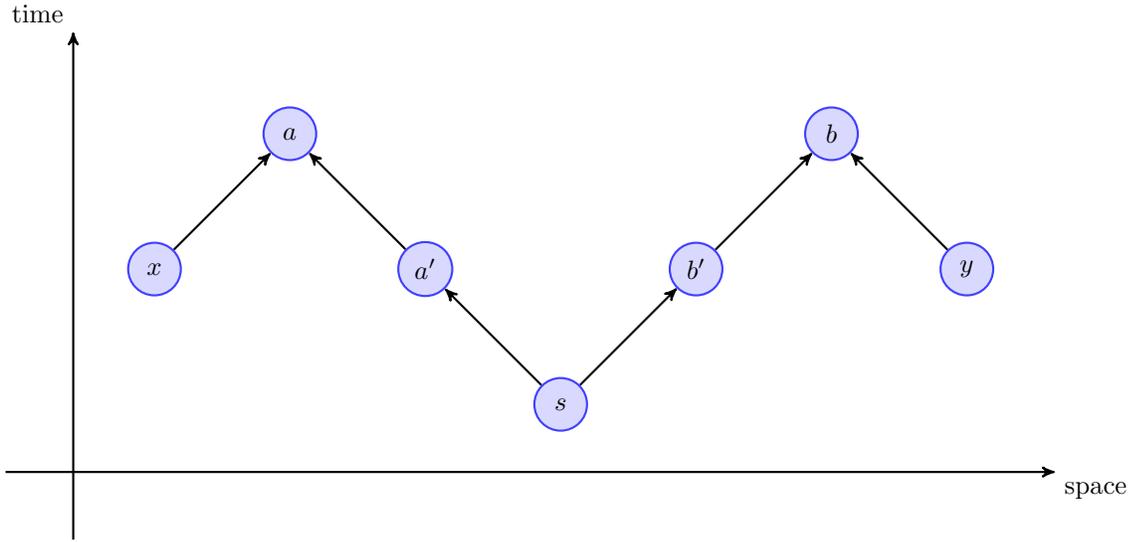
\item The first interpretation connects nicely with the third purpose of this paper: our definitions and results can also be understood as a contribution to the subject of \emph{causal inference}~\cite{Pearl}. Typically, in causal inference one starts with a joint probability distribution of a finite number of random variables and then asks which underlying causal structures, describing which variables influence which other ones, may have generated the given joint distribution. With our definitions, the notion of \emph{hidden variable} or \emph{latent variable}~\cite{Pearl} becomes central to causal inference. The basic idea is as follows: let us assume that the joint probability distribution that we start with represents the distributions of outcomes of probings of a physical system at various points in spacetime, as outlined above. At each of these points, we sample the system, but we do not know whether the information that we obtain by that sampling is a complete representation of the state of the system at that particular point in spacetime. In other words, it may be that there is some ``hidden information'' propagating through that point in spacetime, while not being revealed by the probing. This seems like a reasonable assumption to make when dealing with systems for which we are uncertain whether our theoretical description is complete. We imagine that this applies in particular to biological systems, whose functioning can rarely be comprehensively understood; but it also applies to fundamental physics, as exemplified by Einstein's conviction that quantum mechanics was ``incomplete'' and hidden variables had to be introduced which would describe additional degrees of freedom that were hitherto unknown~\cite{EPR}. To summarize, we believe that any reasonable generalization of Bell scenarios needs to consider hidden variables not only behind events representing sources of (entangled) particles, but behind \emphalt{every} spacetime event relevant to the gedankenexperiment. For the standard Bell scenario of Figure~\ref{Bellscen} and its many-party extensions, we will see in Propositions~\ref{nsbox},~\ref{cbell} and~\ref{bellclassification} that the resulting formalism is equivalent to the conventional one.

Connections between causal inference and Bell's theorem have previously been pointed out by Wood and Spekkens~\cite{WS}. Similar connections have also been explored by Henson, Lal and Pusey~\cite{HLP} and were also hinted at in our previous paper~\cite{Fri}. Most recently, they are being actively investigated by Chaves, Gross and collaborators~\cite{CLG,CLMGJS,CMG,CKBG}.

\end{enumerate}

\makeatletter
\newcommand*{\centerfloat}{%
\parindent \z@
\leftskip \z@ \@plus 1fil \@minus \textwidth
\rightskip\leftskip
\parfillskip \z@skip}
\makeatother
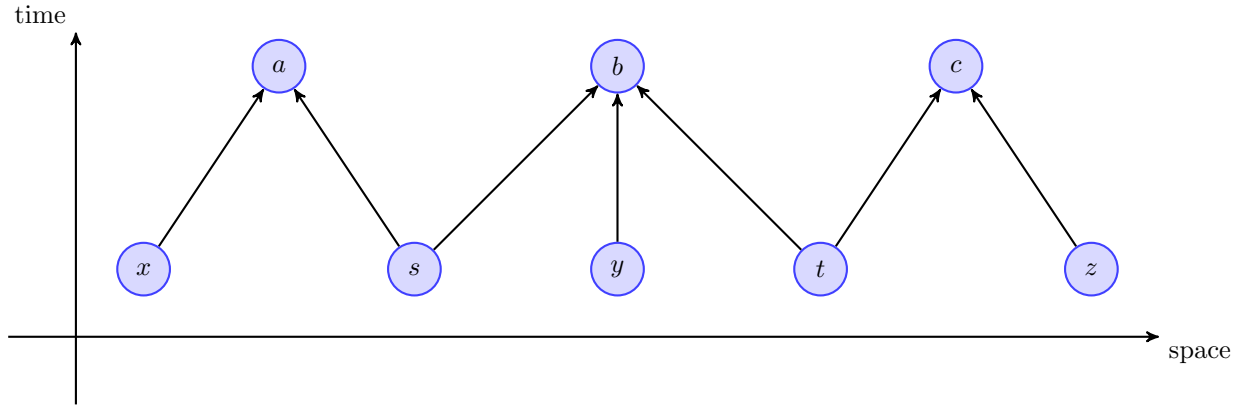
\begin{figure}
\centerfloat
\begin{tikzpicture}[node distance=2.3cm,>=stealth',thick,scale=.9]
\draw[->] (-2,-1) -- (-2,4.5) node[anchor=south east] {time} ;
\draw[->] (-3,0) -- (14,0) node[anchor=north west] {space} ;
\tikzstyle{place}=[circle,thick,draw=blue!75,fill=blue!15,minimum size=7mm]
\node[place] (l) at (3,1) {$s$} ;
\node[place] (la) at (-1,1) {$x$} ;
\node[place] (lb) at (9,1) {$t$} ;
\node[place] (a2) at (1,4) {$a$} ;
\node[place] (b2) at (6,4) {$b$} ;
\node[place] (c) at (11,4) {$c$} ;
\node[place] (lc) at (13,1) {$z$} ;
\node[place] (y) at (6,1) {$y$} ;
\draw[->] (la) -- (a2) ;
\draw[->] (lb) -- (b2) ;
\draw[->] (l) -- (a2) ;
\draw[->] (l) -- (b2) ;
\draw[->] (lb) -- (c) ;
\draw[->] (lc) -- (c) ;
\draw[->] (y) -- (b2) ;
\end{tikzpicture}
\caption{Causal structure of Branciard--Gisin--Pironio's ``bilocality'' scenario~\cite{BGP}.}
\label{bilocal}
\end{figure}

\subsection*{Summary and statement of results.}

Now that the main ideas have been explained, we outline in more detail our definitions and results, the proofs of which will sometimes constitute a technical tour de force. We suspect that some of our results may seem quite intuitive and unsurprising, but, as the technical difficulty of the proofs will demonstrate, they are non-trivial. In fact, despite sustained efforts, we have not achieved perfect rigor in all our arguments; the corresponding arguments are marked as ``\textsc{Sketch of proof}'' rather than ``\textsc{Proof}''. We challenge any skeptical reader to find simpler and completely rigorous proofs.

We start out in Section~\ref{causalsec} by discussing the relevant notion of causal structure, which is given by the concept of \emph{directed acyclic graph}, just as in the subject of Bayesian networks and causal inference~\cite{Pearl}. This graph describes which events have a \emphalt{direct} influence on which other ones through causal links. During the course of our investigations, it will turn out that it is sufficient to know which events have a direct \emphalt{or indirect} influence on which other ones, and it is not necessary to distinguish between a direct and an indirect causal link. In fact, Theorem~\ref{onlyposet} shows that only the \emph{partial order} (or equivalently, \emph{causal set}) induced by the graph needs to be known in order to discuss classical, quantum and other correlations on the given causal structure.

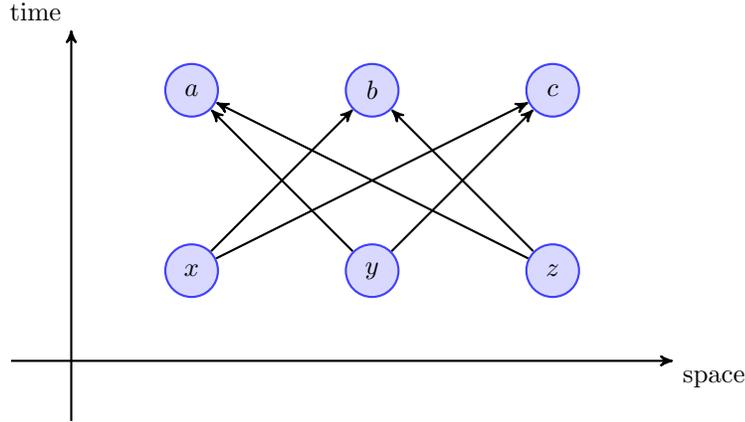
\begin{figure}
\centering
\begin{tikzpicture}[scale=.8,node distance=2.1cm,>=stealth',thick]
\draw[->] (-1,-1) -- (-1,5.5) node[anchor=south east] {time} ;
\draw[->] (-2,0) -- (9,0) node[anchor=north west] {space} ;
\tikzstyle{place}=[circle,thick,draw=blue!75,fill=blue!15,minimum size=7mm]
\node[place] (a) at (1,4.5) {$a$} ;
\node[place] (b) at (4,4.5) {$b$} ;
\node[place] (c) at (7,4.5) {$c$} ;
\node[place] (x) at (1,1.5) {$x$} ;
\node[place] (y) at (4,1.5) {$y$} ;
\node[place] (z) at (7,1.5) {$z$} ;
\draw[->] (x) -- (b) ;
\draw[->] (x) -- (c) ;
\draw[->] (y) -- (a) ;
\draw[->] (y) -- (c) ;
\draw[->] (z) -- (a) ;
\draw[->] (z) -- (b) ;
\end{tikzpicture}
\caption{The triangle scenario from~\cite{Fri} in our present definition. In this two-dimensional drawing, the causal structure may look contrived; but in a three- or four-dimensional spacetime, the events can be arranged in such a way that the causal structure is indeed the one induced from the lightcones, for example by spatially arranging all six events as the vertices of a regular hexagon, while $x$, $y$ and $z$ happen at a suitably earlier time than $a$, $b$ and $c$.}
\label{trianglefig}
\end{figure}

Section~\ref{seccorr} continues by introducing the notion of \emph{correlation} on a given causal structure. Roughly, a correlation is a joint distribution over the outcome variables situated at the events which makes two variables independent as soon as they have disjoint causal past. This formalizes Reichenbach's \emph{common cause principle} ``no correlation without causation''~\cite{cc} in our framework. On a given causal structure, these independence conditions constitute a minimal set of requirements that one would expect an observed joint probability distribution to satisfy, assuming that it arises from that causal structure. This is analogous to the no-signalling condition in a Bell scenario, which is likewise a minimal requirement for a conditional distribution $P(ab|xy)$ to be regarded as potentially physical. As shown in Proposition~\ref{nosig}, our definition essentially reduces to the no-signalling assumption in the case of a Bell scenario.

Then, in Section~\ref{secccorr}, we introduce hidden variables. A given correlation, as in the previous paragraph, is a \emph{classical correlation} if it can be explained in terms of hidden variables being passed along the causal links and being processed at the nodes. This formalizes the above idea of causal inference with hidden variables: one asks which observational results are compatible with the given causal structure, while permitting unobserved information to propagate along the causal links. A given joint distribution is compatible with the causal structure in this sense if and only if it is a classical correlation. It is an essential feature of our definitions and results that we do not restrict a hidden variable to have a finite domain, but allow for it to live on an arbitrary measurable space. The measure-theoretical background necessary for dealing with hidden variables with arbitrary domain is outlined in Appendix~\ref{appsop}. We show in Proposition~\ref{ccorriscorr} that every classical correlation is in fact a correlation. Our further results, in parts quite difficult to prove, are the following: Proposition~\ref{cbell} characterizes classical correlations in a Bell scenario as those whose associated conditional distribution is ``local'' in the conventional sense~\cite{Bellreview}. Theorem~\ref{finitethm} then shows that any classical correlation on any causal structure can be approximated arbitrarily well by other classical correlations for which the hidden variables take on only finitely many values; the proof relies on some auxiliary measure-theoretical results derived in Appendix~\ref{secfinapp}. Towards the end of the section, Conjecture~\ref{detconj}, together with its partial proof, concerns the question whether one can assume without loss of generality that the information processing from incoming to outgoing hidden variables at each event with incoming hidden variables is deterministic. 

Motivated by the question of how to define quantum correlations within our formalism, Section~\ref{seccats} discusses the general properties that a theoretical framework must satisfy in order for a corresponding notion of correlation within that framework to exist. We find that, if a framework for physical theories is given by a symmetric monoidal category $\Cp$ satisfying suitable additional conditions, then there exists a notion of ``$\Cp$-correlation''. We explain how classical correlation are an instance of this general construction. In complete generality, Theorem~\ref{onlyposet} shows that the set of $\Cp$-correlations only depends on the causal structure as a partially ordered set instead of all the details of a directed graph.

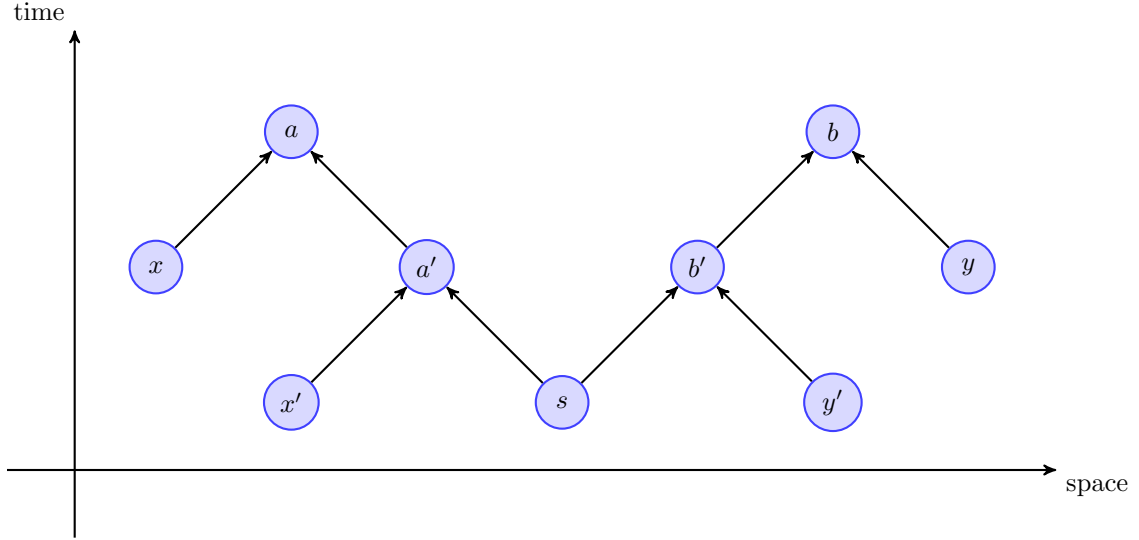
\begin{figure}
\centering
\begin{tikzpicture}[node distance=2.3cm,>=stealth',thick,scale=.9]
\draw[->] (-2.2,-1) -- (-2.2,6.5) node[anchor=south east] {time} ;
\draw[->] (-3.2,0) -- (12.3,0) node[anchor=north west] {space} ;
\tikzstyle{place}=[circle,thick,draw=blue!75,fill=blue!15,minimum size=7mm]
\node[place] (a1) at (3,3) {$a'$} ;
\node[place] (b1) at (7,3) {$b'$} ;
\node[place] (l) at (5,1) {$s$} ;
\draw[->] (l) -- (a1) ;
\draw[->] (l) -- (b1) ;
\node[place] (la) at (-1,3) {$x$} ;
\node[place] (lb) at (11,3) {$y$} ;
\node[place] (xp) at (1,1) {$x'$} ;
\node[place] (yp) at (9,1) {$y'$} ;
\node[place] (a2) at (1,5) {$a$} ;
\node[place] (b2) at (9,5) {$b$} ;
\draw[->] (la) -- (a2) ;
\draw[->] (lb) -- (b2) ;
\draw[->] (a1) -- (a2) ;
\draw[->] (b1) -- (b2) ;
\draw[->] (xp) -- (a1) ;
\draw[->] (yp) -- (b1) ;
\end{tikzpicture}
\caption{Causal structure of the two-party sequential correlation scenario with sequences of two measurements~\cite{GWCAN}.}
\label{seqnl}
\end{figure}

Building on these considerations, we define quantum correlations in Section~\ref{secqc}. We hypothesize in Conjecture~\ref{cisq} that every classical correlation is a quantum correlation, but due to the technical challenges involved, we can only offer a partial proof. Then, Proposition~\ref{bellclassification} shows how our definition of quantum correlation corresponds to the conventional one in the case of Bell scenarios.

Finally, Section~\ref{hbn} discusses another point of view on classical correlations. By analogy with hidden Markov models, we introduce hidden Bayesian networks. The main result here is Theorem~\ref{hbnthm}, which shows that a given joint probability distribution over the outcome variables is a classical correlation if and only if it can be modelled by a hidden Bayesian network.

\subsection*{Terminology and Notation.}

We collect here the relevant notation that we use for graphs $G=(V,E)$. All the graphs $G$ that we consider are \emphalt{directed acyclic graphs} as in Definition~\ref{defdag}, and hence we usually omit the qualifiers ``directed acyclic'' and simply refer to $G$ as a \emphalt{graph}.

While we always think of the nodes $v\in V$ as events in spacetime, we mostly call them ``nodes'' rather than ``events'', and regard these two terms as essentially synonymous. Similarly, we usually speak of ``edges'' $e\in E$ rather than ``causal links'', although we think of the edges as causal links along which physical systems propagate. (With the exception of the edges of the hidden Bayesian networks of Section~\ref{hbn}, which we do not interpret like this.)

For nodes $v\in V$ and edges $e\in E$, we put
\begin{alignat*}{2}
\src{e} & \defin \textrm{source node of } e && \\
\tar{e} & \defin \textrm{target node of } e && \\
\pa{v} & \defin \textrm{parent nodes of } v && = \{ u\in V \:|\: u\to v\} \\
\ch{v} & \defin \textrm{children nodes of } v && = \{ w \in V \:|\: v\to w \} \\ 
\pst{v} & \defin \textrm{causal past of } v && = \{ u\in V \:|\: u\to\ldots\to v \} 
\end{alignat*}
We can also apply each of these functions to sets of edges (respectively nodes), in which case we put $\pa{U}\defin \bigcup_{v\in U}\pa{v}$, and similarly for the other functions. See Section~\ref{causalsec} for more detail on all of this.

In order to keep the notation from being too confusing, we distinguish between a node and the corresponding outcome variable it carries. If $v\in V$ is a node, then the outcome variable at $v$ is denoted $o[v]$, and $P(o[v])$ is its distribution. This notation is unfortunately somewhat unusual, since in the Bell scenario of Figure~\ref{Bellscen} we need to denote e.g.~the outcome of Alice by ``$o[a]$'' instead of the conventional ``$a$'', but it will help to keep things clear.

\subsection*{Disclaimer.}

The definitions proposed in the present paper are not intended to be final or conclusive in any sense. Rather, they constitute only one possible approach for generalizing Bell scenarios to arbitrary causal structures, and we expect that other definitions are possible and potentially useful~\cite{HLP}. It may also be desirable to have definitions which even allow for the consideration of indefinite causal structure~\cite{OCB}, although we do not know how to achieve this. Further investigations will have to tell which definitions are the most adequate and useful; the only thing that seems certain at the present time is that Bell scenarios constitute only a small subclass of all conceivable scenarios, and that there are many other interesting ones that one can consider~\cite{Fri,HLP}.

A potentially problematic aspect of our approach is this: in order to observe the joint probability distribution of outcomes of probings in an experimental situation, one needs to gather statistics, i.e.~repeat the experiment many times \emphalt{with independent trials}. This independence assumption of trials goes somewhat against the spirit of allowing hidden variables wherever conceivable, and might be hard to justify in certain applications, e.g.~in neurophysiology, where the neurons may also store hidden information between the trials. At present, we do not know under which conditions this independence assumption can be relaxed, but the results of Bierhorst~\cite{Bierhorst} for the case of the Bell scenario are encouraging.

\newpage
\section{What is causal structure?}
\label{causalsec}

Our goal is to study correlations between different observations conducted on a collection of systems and to determine which causal structures between these observations are compatible with these correlations. So we first need to discuss what the term ``causal structure'' actually means. Generally, a causal structure is defined on a collection of events $V$, and the function of a causal structure is to clarify which events may \emph{influence} which other ones. Whenever event $v$ may influence event $w$, we write $v\to w$. Upon taking the events to be nodes and regarding the influence relations $v\to w$ as edges, we therefore obtain a \emph{directed graph} $G=(V,E)$, i.e.~a set of nodes $V$ together with a set of directed edges $E\subseteq V\times V$, which becomes the mathematical model of a causal structure. If there is no causal link from $v$ to $w$, then we write $v\not\to w$. We will always assume that $V$ is finite, and then so is $E$; it is probably possible to relax this assumption, albeit at the cost of additional technical challenges. Any edge $e\in E$ is necessarily of the form $v\to w$ for a unique pair of nodes $v,w\in V$, which we call the \emph{source} and \emph{target} of $e$ and also denote by $\src{e}$ and $\tar{e}$. For $v\in V$, we write
\[
\inc{v} \defin \set{ e\in E \: | \: \tar{e} = v } 
\]
for the collection of \emph{incoming edges} at $v$, and similarly
\[
\out{v} \defin \set{ e\in E \: | \: \src{e} = v }
\]
for the \emph{outgoing edges} at $v$. The upcoming acyclicity condition of Definition~\ref{defdag} will guarantee that these two sets are disjoint. In Section~\ref{hbn}, we will use similar concepts for nodes: the set
\[
\pa{v} \defin \set{ u\in V \: | \: u\to v }
\]
is the collection of \emph{parents} of $v$, while
\[
\ch{v} \defin \set{ w\in V \: | \: v\to w }
\]
are its \emph{children}.

We typically think of the events or nodes of $G$ as points in spacetime, and there is an edge or causal link $v\to w$ if $w$ lies in the future lightcone of $v$. However, our mathematical results are completely general and do not depend on this interpretation. In particular, we are confident that other ways for obtaining candidate causal structures are feasible and interesting, in particular in applications to mathematical modelling in other fields like biology or climate science. In these situations, it will often be the case that it is not known which events can influence which other ones, and therefore the causal structure is not known in advance. Applying our formalism to these cases is nevertheless possible by first choosing some ``trial'' causal structures and then asking, for each of these, whether it is compatible with the given observations or not. We do not subscribe to any particular method for finding such trial causal structures; doing so is up to the user of our formalism. It should be an interesting challenge to come up with algorithms for doing so, possibly in analogy with the standard causal inference algorithms commonly used~\cite[p.~50ff]{Pearl}.

In general, even if $u\not\to w$, it may be possible that there is an \emph{indirect} causal link from $u$ to $w$ proceeding through other nodes:

\begin{defn}
\label{dirpath}
A \emph{directed path} from $u \in V$ to $w \in V$ is a finite sequence of nodes $v_1, \dotsc, v_n \in V$ with $u \to v_1 \to \dotsb \to v_n \to w$.
\end{defn}

We use the shorthand $u \rsa w$ for the existence of a directed path from $u$ to $w$, and $u \not\rsa w$ for the absence of any such directed path. In the conventional formalism of Bayesian networks~\cite{Pearl}, no information can propagate along indirect causal links since the intermediate nodes ``screen off'' any potential indirect influence of $u$ on $w$, and hence the outcome variable $o[u]$ may become conditionally independent of $o[w]$ given $o[v]$ (whether this happens depends on the overall structure of the graph). For us, this is no longer the case: the outcome variable $o[v]$ living on an intermediate node $v$ constitutes merely a probing of the environment without the guarantee that it faithfully represents all information propagating through $v$, and hence we cannot expect $o[v]$ to screen off any correlations between $o[u]$ and $o[w]$.

The additional property enjoyed by ``$\rsa$'' is \emph{transitivity}:
\[
u\rsa v,\quad v\rsa w \quad \Longrightarrow \quad u\rsa w.
\]
This makes sense if one interprets ``$u\rsa v$'' as ``$u$ is in the causal past of $v$''. So the information carried by $\rsa$ is exactly that of a \emph{partially ordered set}, which for us plays a r{\^o}le similar to that of a \emph{causal set}~\cite{Henson}. Due to the frequent occurrence of the concept, we introduce a special notation for the causal past,
\[
\pst{v} \defin \set{ v } \cup \set{ \: u\in V \: | \: u\rsa v \: }.
\]
It is important to note that we consider every node itself to be a member of its causal past.

A condition that any reasonable causal structure should satisfy is that there are no causal loops, i.e.~that no event can influence itself in any non-trivial way:

\begin{defn}[{e.g.~\cite{Pearl}}]
\label{defdag}
A \emph{directed acyclic graph} is a directed graph with $v \not\rsa v$ for all $v \in V$.
\end{defn}

Note that this definition is sensible only if we do not allow paths of length zero in Definition~\ref{dirpath}.

For example, Figures~\ref{Bellscen}--\ref{seqnl} depict directed acyclic graphs. Acyclicity means that there is no directed sequence of edges which returns to its starting point. In particular, no node has a self-loop.

The particular causal structure considered by Bell~\cite{Bell} is the one of Figure~\ref{Bellscen}. It is worth noting that the five nodes of that figure are normally considered to play very different roles: $s$ represents a \emph{source} emitting particle pairs or other physical systems, one in each direction. It is usually not considered to yield an outcome $o[s]$, although it almost certainly does in any actual experiment, where the outcome $o[s]$ in any given run will be one of ``successfully produced an entangled particle pair'' or ``failed to produce an entangled particle pair''. The nodes $x$ and $y$ stand for a \emph{choice of measurement} randomly selected by an agent (or by an apparatus, in most actual experiments) with values $o[a]$ and $o[b]$. Finally, the events $a$ and $b$ represent the execution of these chosen measurements on the systems emitted by the source which produce outcomes $o[a]$ and $o[b]$. However, this categorization of the events into three different kinds is completely superfluous: one can consider all five variables $o[s]$, $o[x]$, $o[y]$, $o[a]$, $o[b]$ to be outcomes of observations, which are (probabilistic) functions of the values of the incoming hidden variables propagating along the incoming causal links (if any), and emitting hidden variables along the outgoing causal links (if any). Here, the ``hidden variables'' may be classical, quantum or other kinds of systems; we will show later that in this picture, our corresponding definitions of classical and quantum correlations coincide precisely with the usual ones~\cite{Bellreview}.

So, generalizing and reinterpreting some of Bell's work is what we set out to do in this paper. The questions that naturally arise are of this kind: under which conditions can a given correlation be explained in terms of hidden variables on a given causal structure? We are not able to answer this question in this paper and we suspect that it is in general very difficult, but we provide a set of basic definitions and results for considering it.

\newpage
\section{Correlations}
\label{seccorr}

So what are these ``correlations'' between observations which we were talking about? As mentioned before, we imagine that at every node $v\in V$, there is an outcome $o[v]$ describing the result of the corresponding observation and taking values in a finite set of possible values denoted $O_v$. It takes a concrete value in each run of the experiment, but possibly not the same value in each run, so that it needs to be described in terms of a probability distribution over outcomes. We also write $o[V]$ for the vector or tuple $(o[v])_{v\in V}$ consisting of all the outcomes $o[v]$ for all nodes $v\in V$. In each run of the experiment, we obtain one concrete tuple of values for $o[V]$, so that we obtain a probability distribution $P(o[V])$, which is a joint distribution of all the individual outcome variables $o[v]$. If $U\subseteq V$ is any subset of the nodes, then we similarly use notation like $o[U]$ and $P(o[U])$ for the corresponding tuples and their observed distribution.

Now, under which condition is a given joint distribution $P(o[V])$ explainable in terms of hidden variables on a given causal structure? This is a very difficult question in general and we are very far from having an answer, but there is an obvious necessary condition: if two subsets $V_1,V_2\subseteq V$ have distinct causal past, then $o[{V_1}]$ and $o[{V_2}]$ should be independent. More generally:

\begin{defn}
\label{defcor}
$P(o[V])$ is a \emph{correlation} if for any collection of subsets $V_1,\ldots,V_m\subseteq V$ with disjoint causal past, i.e.~$\pst{V_i} \cap \pst{V_j} = \emptyset$ for $i\neq j$, the distribution factors as
\beq
\label{correlation}
P(o[V_1]\ldots o[V_m]) = P(o[{V_1}])\cdots P(o[{V_m}]).
\eeq
\end{defn}

This is the basic necessary condition that we will use. From now on, we restrict our considerations to those joint distributions $P(o[V])$ that are correlations. No other joint distributions of outcomes will be considered on a given causal structure. When we define classical correlations, it will be our first result (Proposition~\ref{ccorriscorr}) to show that every classical correlation is actually a correlation, and similarly in the quantum case.

In deriving the concrete form of the constraints~\eqref{correlation} for a given (candidate) causal structure, it is sufficient to postulate~\eqref{correlation} only for $m=2$, i.e.~when one only considers two sets $U$ and $W$ with disjoint causal past, in which case the condition is
\beq
\label{bicorrelation0}
P(o[U] o[W] ) = P(o[U])\, P(o[W]).
\eeq
The reason for the equivalence is that one can derive~\eqref{correlation} from this assumption recursively. For example, for three sets $V_1$, $V_2$, and $V_3$ with disjoint causal past, one gets
\begin{align*}
P(o[V_1] o[V_2] o[V_3]) & = P(o[V_1] o[{V_2\cup V_3}]) \stackrel{\eqref{bicorrelation0}}{=} P(o[{V_1}]) \, P(o[{V_2\cup V_3}]) \\
&= P(o[{V_1}]) \, P(o[{V_2}] o[{V_3}]) \stackrel{\eqref{bicorrelation0}}{=} P(o[{V_1}])\, P(o[{V_2}]) \, P(o[{V_3}]).
\end{align*}
Here, the second equation uses the fact that $V_1$ and $V_2\cup V_3$ have disjoint causal past, which follows from the assumption that all the $V_i$ have pairwise disjoint causal past.

Moreover, in~\eqref{bicorrelation0} it is enough to consider those $U$ and $W$ which are \emph{maximal} in the sense that there is no node outside of these sets which could be added to either of them, while still preserving the condition of disjoint past. The reason for this equivalence is that for any pair $U,W$, there is a maximal one $U',W'$ with $U\subseteq U'$ and $W\subseteq W'$. Now the assumption
\[
P(o[U'] o[W']) = P(o[{U'}])\, P(o[{W'}])
\]
implies the desired equation~\eqref{bicorrelation} upon taking the marginal of both sides over all those variables which are contained in $U'$ or $W'$, but not in $U$ or $W$. Note that maximality requires in particular that $\pst{U}=U$ and $\pst{W}=W$, since otherwise $U$ or $W$ could be enlarged. Thereby, we have shown:

\begin{prop}
\label{corrprop}
$P(o[V])$ is a correlation if and only if for any maximal pair of subsets $U,W\subseteq V$ with $\pst{U}\cap\pst{W}=\emptyset$, one has
\beq
\label{bicorrelation}
P(o[U] o[W]) = P(o[U])\, P(o[W]).
\eeq
\end{prop}

For any given causal structure, it is instructive to write down a minimal set of conditions on $P(o[V])$ which guarantee that it is a correlation. We now carry this out in the familiar case of a Bell scenario. Instead of considering only the $2$-party scenario of Figure~\ref{Bellscen}, we generalize it to the \emph{$n$-party Bell scenario} in which the two ``arms'' $\{a,x\}$ and $\{b,y\}$ are replaced by an arbitrary number of arms $\{a_1,x_1\},\ldots,\{a_n,x_n\}$. More precisely, the $n$-party Bell scenario has $2n+1$ nodes: a ``source'' $s$, ``settings'' $x_1, \dotsc, x_n$ and ``outcomes'' $a_1, \dotsc, a_n$. For each party $i$ we have $s \to a_i$ and $x_i \to a_i$, and there are no other edges. Our first result about the $n$-party Bell scenario is that correlations on this causal structure can be characterized like this:

\begin{prop}
\label{nsbox}
In the $n$-party Bell scenario,  $P$ is a correlation if and only if it satisfies the \emph{free will equation}
\beq
\label{freewill}
P(o[{x_1}\ldots {x_n} s]) = P(o[{x_1}])\:\cdots\, P(o[{x_n}]) \, P(o[s])
\eeq
and the \emph{no-signalling equation}
\beq
\label{nosig}
P(o[{a_1}\ldots \bcancel{{a_i}} \ldots {a_n} {x_1}\ldots {x_n} s]) = P(o[{x_i}]) P(o[{a_1}\ldots \bcancel{{a_i}}\ldots {a_n} {x_1}\ldots \bcancel{{x_i}}\ldots {x_n} s])
\eeq
for every party $i=1,\ldots,n$.
\end{prop}

If one makes use of~\eqref{freewill}, then the no-signalling equation~\eqref{nosig} can be seen to be equivalent to the more conventional form
\[
P(o[{a_1}\ldots \bcancel{{a_i}} \ldots {a_n}] \,|\, o[{x_1}\ldots {x_n} s]) = P(o[{a_1}\ldots \bcancel{{a_i}}\ldots {a_n}] \,|\, o[{x_1}\ldots \bcancel{{x_i}}\ldots {x_n} s]),
\]
which is usually interpreted as stating that the conditional probability on the left is independent of the particular value of $o[{x_i}]$. Also, it is usually assumed in the Bell scenario that the source $s$ does not have an outcome; in our formalism, this corresponds to taking the source's outcome $o[s]$ to be always the same. In an actual experiment, however, the source typically does have an outcome, corresponding to whether it successfully produced the required pair of entangled particles or not. In this sense, our approach automatically treats a Bell scenario in a more realistic manner than is usually done. Another difference is that a correlation in our approach does include chosen probabilities $P(o[{x_i}])$ for the ``settings'' $o[{x_i}]$. Again, this seems more realistic to us: in an actual experiment, the $o[{x_i}]$ have particular distributions, and these are even relevant in practical information processing applications like device-independent randomness expansion, where strong limitations on the $P(o[{x_i}])$ need to be imposed~\cite{VVrand}. Another example is device-independent quantum key distribution~\cite{VVqkd}, where protocol performance also depends on the $P(o[{x_i}])$, even in such a way that highly non-uniform distributions are optimal.

\begin{proof}
We begin with the ``only if'' part. To establish~\eqref{freewill} from the assumption that $P(o[V])$ is a correlation, use
\[
V_1 = \{x_1\},\quad\ldots,\quad V_n = \{x_n\},\quad V_{n+1} = \{s\}
\]
in~\eqref{correlation}. To establish~\eqref{nosig}, put similarly
\[
V_1 = \{x_i\},\qquad V_2 = \{a_1\ldots\bcancel{a_i}\ldots a_n x_1 \ldots \bcancel{x_i} \ldots x_n s\}.
\]
In both cases, it is straightforward to check that the condition of having disjoint causal past holds.

For the ``if'' part, suppose~\eqref{freewill} and~\eqref{nosig} hold. Thanks to Proposition~\ref{corrprop}, it suffices to check the correlation condition for two maximal disjoint-past sets $U$ and $W$. By $\pst{U}\cap \pst{W}=\emptyset$, at least one of $\pst{U}$ and $\pst{W}$ does not contain $s$; without loss of generality, we take this to be $U$. Therefore, by maximality, $s$ must be contained in $W$. Then since $U$ and $W$ have disjoint causal past, $U$ cannot contain $s$ or any $a_i$, so $U = \set{ x_i | i\in I }$ for some $I \subseteq \{1, \dotsc, n\}$. Again by maximality, we conclude that $W = \set{a_i, x_i | i\not\in I } \cup \set{s}$. If we write $\bar{I} \defin \{1,\dotsc,n\}{\setminus} I$, then the equation that we have to derive from~\eqref{freewill} and~\eqref{nosig} is that
\[
P(o[{a_{\bar{I}}}  {x_I} {x_{\bar{I}}} s] ) = P(o[{x_I}])\, P(o[{a_{\bar{I}}}{x_{\bar{I}}}s]),
\]
for an arbitrary set of parties $I\subseteq\{1,\ldots,n\}$. This equation can be interpreted as the analogue of~\eqref{nosig} in which the individual party $i$ is replaced by the set $I$. That this indeed follows from the assumptions is analogous to a well-known fact in the conventional formalism of no-signaling boxes, where no-signaling across a bipartition $I$ vs $\bar{I}$ follows from no-signaling with respect to each party individually~\cite{BLMPPR}. With our definitions, this derivation can be carried out by induction on the size of $I$: if $I$ contains just one party, then our target equation coincides with the assumption~\eqref{nosig}. For the induction step, we assume that the equation holds for $I$ and prove that it also holds for $I\cup\set{j}$ in place of $I$, where $j$ is an arbitrary party not in $I$. This works as follows:
\begin{align*}
P(o[a_{\bar{I}{\setminus}\{j\}} x_{I\cup\{j\}} x_{\bar{I}{\setminus}\{j\}} s]) & = P(o[{a_{\bar{I}{\setminus}\{j\}}}{x_I} {x_{\bar{I}}} s]) \\
& = P(o[{x_I}])\, P(o[{a_{\bar{I}{\setminus}\{j\}}} {x_{\bar{I}}} s]) \\
& \stackrel{\eqref{nosig}}{=} P(o[{x_I}])\, P(o[{x_j}])\, P(o[{a_{\bar{I}{\setminus}\{j\}}} {x_{\bar{I}{\setminus}\{j\}}} s]) \\
& \stackrel{\eqref{freewill}}{=} P(o[{x_{I\cup\{j\}}}])\, P(o[{a_{\bar{I}{\setminus}\{j\}}} {x_{\bar{I}{\setminus}\{j\}}} s]),
\end{align*}
as was to be shown. Here, the second step uses the induction assumption, marginalized over $o[a_j]$. The third and fourth steps use similar marginalizations of~\eqref{nosig} and~\eqref{freewill}.
\end{proof}

Finally, it is worth noting that the set of all correlations on a given graph $G$ does not depend on all details of the graph itself, but only on the partial order ``$\rsa$'' induced by $G$, as discussed in Section~\ref{causalsec}. The reason is that in order to apply Definition~\ref{defcor}, all that one needs to know about the causal structure is which events lie in the causal past of which other ones.

\newpage
\section{Classical correlations}
\label{secccorr}

As discussed in the introduction and the previous section, we want to assume that every node $v\in V$ carries a random variable $o[v]$, which represents the outcome of some probing of spacetime conducted at $v$. The basic assumption that we make is that these probings may not constitute a complete and faithful representation of the physical systems propagating through these events, and in addition the observations may modify those systems in an unknown and arbitrary manner. More concretely, it could be the case that there are additional ``hidden variables'' necessary to describe the systems under consideration, and the observations only reveal partial information about the values of these hidden variables. We believe that this is a hypothesis which naturally arises in many situations in the sciences: for example, we can measure action potentials in a biological neural network and try to infer the causal structure from the correlations behind these action potentials, but how would we know that there is no other form of information being transmitted between the neurons? If we measure a certain number of meteorological parameters in the atmosphere at different points in space and time, how can we be sure that we are not missing relevant information like the concentration of a certain aerosol, which may have a crucial effect on cloud formation? For these reasons, we allow for the propagation of hidden information and for the possibility that the observations modify this hidden information in an unknown manner.

There are two ways to make this idea precise:
\begin{itemize}
\item The hidden variables live on the \emphalt{edges} of the graph.

This is the approach that we pursue in this section. The local structure of a node $v$ is illustrated in Figure~\ref{localhv}: the ingoing edges $e\in E$ carry hidden variables $\lambda_e$, which determine the outcome $o[v]$ and the outgoing hidden variables via a certain conditional distribution $P(o[v]\lambda_{\out{v}}|\lambda_{\inc{v}})$\footnote{Here, we also write $\lambda_{\out{v}}$ as shorthand for the tuple of variables $(\lambda_e)_{e\in\out{v}}$ etc.}. In terms of a computational interpretation, the node $v$ operates like an \emph{information processing gate} taking the incoming hidden variables $\lambda_{\inc{v}}$ and turning them into outgoing hidden variables $\lambda_{\out{v}}$ together with an outcome $o[v]$. In terms of a physical interpretation, we think that it is appropriate to have the hidden variables live on the edges: after all, a hidden variable is supposed to represent the state of a physical system, and since we interpret our causal structures in terms of spacetime pictures, these physical systems should be described by their one-dimensional worldlines. Since we allow the systems to be modified in an arbitrary way by the observations, the conditional distribution $P(o[v]\lambda_{\out{v}}|\lambda_{\inc{v}})$ may be completely arbitrary. In particular, the outgoing systems may be of a different kind than the incoming ones, and we do not make any sort of ``noninvasiveness'' assumption~\cite{LG}, and hence our formalism does not comprise Leggett-Garg scenarios.
\item The hidden variables live on the \emphalt{nodes} of the graph. 

We develop this approach in Section~\ref{hbn} and show it to be equivalent to the previous one. While it is less susceptible to a physical or information processing interpretation, it is very natural from the mathematical point of view and fits in nicely with the standard approach to causal inference in terms of Bayesian networks~\cite{Pearl}. It captures the idea of assuming that there may be hidden information ``behind'' each node.
\end{itemize}

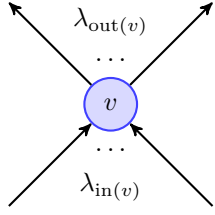
\begin{figure}
\begin{tikzpicture}[node distance=2.1cm,>=stealth',thick]
\tikzstyle{place}=[circle,thick,draw=blue!75,fill=blue!15,minimum size=7mm]
\node[place] (node) at (0,0) {$v$} ;
\node [above right of=node] (out2) {} ;
\node [above left of=node] (out3) {} ;
\node [below right of=node] (in1) {} ;
\node [below left of=node] (in2) {} ;
\draw[->] (node) -- (out2) ;
\draw[->] (node) -- (out3) ;
\draw[<-] (node) -- (in1) ;
\draw[<-] (node) -- (in2) ;
\node at (0,1.1) {$\lambda_{\out{v}}$} ;
\node at (0,.6) {$\ldots$} ;
\node at (0,-1.1) {$\lambda_{\inc{v}}$} ;
\node at (0,-.6) {$\ldots$} ;
\end{tikzpicture}
\caption{Intended local structure of a node $v$.}
\label{localhv}
\end{figure}

We now develop the first approach in some detail.

\begin{defn}
\label{defccorr}
$P(o[V])$ is a \emph{classical correlation} if for every $e\in E$ there is a hidden variable $\lambda_e$ given in terms of a probability space $(X_e,\Sigma_e)$, and for every $v\in V$ there is a conditional distribution
\beq
\label{gate}
P(o[v] \lambda_{\out{v}} | \lambda_{\inc{v}} )
\eeq
such that
\beq
P(o[V]) = \int_{\lambda_E} \prod_{v\in V} P(o[v] \lambda_{\out{v}} | \lambda_{\inc{v}}) .
\label{ccorr}
\eeq
\end{defn}

There are many possible variations on this definition having to do with the precise mathematical definitions used in the case that the hidden variables may take infinitely many values, which we explicitly allow. The definitions that we work with are explained in Appendix~\ref{appsop}; we will not consider other possible definitions in this work. We suspect that the set of classical correlations does not depend on which definitions are used, but we do not have a proof of this. The readers not interested in non-discrete hidden variables may want to restrict themselves to taking the $X_e$ to be finite or countable sets.

Some further explanation of~\eqref{ccorr} is in order. For any concrete value of the tuple of outcomes $o[V]$, the right-hand side evaluates to a non-negative number as follows. Any $P(o[v] \lambda_{\out{v}} | \lambda_{\inc{v}})$ is a $\lambda_{\inc{v}}$-dependent subnormalized measure on the product space $X_{\out{v}} = \prod_{e\in\out{v}} X_e$. The product in~\eqref{ccorr} hence yields a measure on the total product space $X_E = \prod_{e\in E} X_e$. Integrating this measure results in a number, corresponding to its normalization, which is the joint probability of getting the joint outcome $o[V]$. Strictly speaking, making sense of this product of conditional distributions and the integral requires the methods of Section~\ref{seccats}, since it is not obvious how to show that the integrand does indeed define a measure. We ignore this technical point for the time being.

As the terminology suggests, a classical correlation is indeed a correlation in the sense of Section~\ref{seccorr}.

\begin{prop}
\label{ccorriscorr}
If $P(o[V])$ is a classical correlation, then it is a correlation.
\end{prop}

\begin{proof}
It is already non-trivial to verify that $P(o[V])$ is actually a probability distribution, i.e.~satisfies the normalization equation $\sum_{o[V]} P(o[V]) = 1$.

We start by showing that if $U\subseteq V$ is some set of nodes with $\pst{U}=U$, then the marginal $P(o[U])$ of~\eqref{ccorr} can be computed like this:
\beq
\label{cmarginal}
P(o[U]) = \int_{\lambda_\out{U}} \prod_{v\in U}  P(o[v] \lambda_{\out{v}} | \lambda_{\inc{v}}) .
\eeq
Here, $\out{U} = \bigcup_{v\in U} \out{v}$, so that the integration variables $\lambda_{\out{U}}$ comprise all those $\lambda_e$ with $\src{e}\in U$. This formula should make intuitive sense: in order to determine the distribution of outcomes in $\pst{U}=U$, it should be sufficient to consider only the nodes and edges in $U$ itself, since no other events or hidden variables can influence the variables in $U$.

We prove~\eqref{cmarginal} by induction on the size of $V{\setminus} U$. The base case is $U=V$, for which the formula is precisely the definition~\eqref{ccorr}. For the induction step, we need to prove that the formula holds for a given $U$ with $\pst{U}=U$, when assuming that it holds for all bigger $U'$ in place of $U$ which likewise satisfy $\pst{U'}=U'$. For the given $U$, we can choose a node $w$ with $\pa{w}\subseteq U$; upon putting $U'\defin U\cup\{w\}$, it is then guaranteed that $\pst{U'}=U'$ as well. We therefore obtain, by the induction assumption,
\begin{align*}
P(o[U]) & = \sum_{o[w]} P(o[U] o[w]) = \sum_{o[w]} \,\int_{\lambda_{\out{U}\cup\out{w}}} P(o[w] \lambda_{\out{w}} | \lambda_{\inc{w}}) \prod_{v\in U} P(o[v] \lambda_{\out{v}} | \lambda_{\inc{v}}) \\
& = \int_{\lambda_{\out{U}}} \int_{\lambda_{\out{w}}} \sum_{o[w]} P(o[w] \lambda_{\out{w}} | \lambda_{\inc{w}}) \prod_{v\in U} P(o[v] \lambda_{\out{v}} | \lambda_{\inc{v}}) \\
& = \int_{\lambda_{\out{U}}}  \left(\prod_{v\in U} P(o[v] \lambda_{\out{v}} | \lambda_{\inc{v}})\right) \,\int_{\out{w}} \sum_{o[w]} P(o[w]\lambda_{\out{w}} | \lambda_{\inc{w}}) \\
& = \int_{\lambda_{\out{U}}}  \prod_{v\in U} P(o[v] \lambda_{\out{v}} | \lambda_{\inc{v}}) .
\end{align*}
Here, in the last step we have used the normalization of probability for the conditional distribution $P(o[v]\lambda_{\out{v}}|\lambda_{\inc{v}})$, which states precisely that the integral-sum on the right evaluates to $1$.

This completes the proof of~\eqref{cmarginal}. In the special case $U=\emptyset$, the right-hand side of~\eqref{cmarginal} is an empty product and therefore evaluates to $1$, which proves that $P(o[V])$ is indeed a probability distribution.

To keep the notation simple, we prove the ``binary'' correlation equation~\eqref{bicorrelation} with the assumption $\pst{U}=U$ and $\pst{W}=W$, which we already know to suffice. Then we have $\pst{U\cup W} = \pst{U} \cup \pst{W} = U\cup W$, and therefore we obtain by~\eqref{cmarginal},
\begin{align*}
P(o[U] o[W]) = & \int_{\lambda_{\out{U\cup W}}} \prod_{v\in U\cup W} P(o[v] \lambda_{\out{v}} | \lambda_{\inc{v}}) \\
 = & \int_{\lambda_{\out{U}}} \int_{\lambda_{\out{W}}} \left( \prod_{v\in U} P(o[v] \lambda_{\out{v}} | \lambda_{\inc{v}}) \right) \cdot\left( \prod_{v\in W} P(o[v] \lambda_{\out{v}} | \lambda_{\inc{v}}) \right) \\
 = & \left( \int_{\lambda_{\out{U}}} \prod_{v\in U} P(o[v] \lambda_{\out{v}} | \lambda_{\inc{v}}) \right) \cdot \left( \int_{\lambda_{\out{W}}} \prod_{v\in W} P(o[v] \lambda_{\out{v}} | \lambda_{\inc{v}}) \right) \\
 = & P(o[U]) \, P(o[W]),
\end{align*}
as was to be shown. 
\end{proof}

In the case of Bell scenarios, classical correlations correspond exactly to those conditional probability distributions that are usually called ``local''~\cite{Bellreview}:

\begin{prop}
\label{cbell}
In the $n$-party Bell scenario, $P$ is a classical correlation if and only if it satisfies the free will equation~\eqref{freewill}, and the corresponding conditional probabilities are of the form
\beq
\label{lcausal}
P(o[a_1\ldots a_n]|o[x_1 \ldots x_n s]) = \int_\lambda P(o[a_1]|o[x_1]\lambda)\cdots P(o[a_n]|o[x_n]\lambda) P(\lambda|o[s])
\eeq
for some single hidden variable $\lambda$ living on some measurable space $(X,\Sigma)$.
\end{prop}

Here, we implicitly also assume that only those values of $o[x_i]$ and $o[s]$ are considered for which $P(o[x_i]) > 0$ and $P(o[s]) > 0$, since otherwise conditioning on these variables would not make sense.

\begin{proof}
We begin with the ``only if'' part. To establish~\eqref{lcausal} for a given classical correlation, we define the conditional distributions occurring on the right-hand side of~\eqref{lcausal} by identifying $\lambda\defin \lambda_{\out{s}}$ and putting 
\begin{align*}
P(\lambda|o[s]) &\defin \frac{P(o[s]\lambda_{\out{s}})}{P(o[s])},\\
P(o[a_i]|o[x_i]\lambda_{\out{s}}) &\defin \int_{\lambda_{x_i\to a_i}} \frac{P(o[a_i]|\lambda_{s\to a_i}\lambda_{x_i\to a_i})P(o[x_i]\lambda_{x_i\to a_i})}{P(o[x_i])}
\end{align*}
where the (conditional) distributions in the numerators of the right-hand sides are all part of the given data which realizes $P$ as a classical correlation, while those in the denominators are marginals thereof. Using the free will equation~\eqref{freewill}, we have
\beq
  P(o[a_1\ldots a_n] | o[x_1 \ldots x_n s]) = \frac{P(o[a_1 \ldots a_n x_1 \ldots x_n s])}{P(o[x_1])\dotsm P(o[x_n])P(o[s])}
  \label{expandedcond}
\eeq
which, when the numerator is expanded using~\eqref{ccorr}, gives \eqref{lcausal}.

Now for the ``if'' part. Suppose~\eqref{freewill} and~\eqref{lcausal} hold. We establish that this is a classical correlation as follows. Let each $\lambda_{s\to a_i}$ be a copy of $\lambda$. Let $P(s\lambda_\out{s})\defin P(\lambda|s)\, P(s)$. We take $\lambda_{x_i\to a_i}$ to be $o[x_i]$ itself, meaning that $P(o[x_i]\lambda_{x_i\to a_i}) \defin P(o[x_i])\cdot\delta_{o[x_i],\lambda_{x_i\to a_i}}$. Finally, let $P(o[a_i]|\lambda_{x_i\to a_i}\lambda_{s\to a_i}) \defin P(o[a_i]|o[x_i]\lambda)$. Multiplying~\eqref{lcausal} by~\eqref{freewill} then gives~\eqref{ccorr}.
\end{proof}

A similar argument will work in the quantum case; see Proposition~\ref{bellclassification}.

Bell scenarios have the desirable feature that the resulting conditional distributions of the form~\eqref{lcausal} form a convex polytope known as the \emph{local polytope}~\cite{Bellreview}. Checking whether a given correlation is classical then simply boils down to asking whether the associated conditional distribution lies in the local polytope. However, for general causal structures, determining whether a given correlation is classical is a very difficult problem. A good example is the triangle scenario displayed in Figure~\ref{trianglefig}, which has been studied in~\cite{Fri,HLP,CLG} for the special case that the outcomes $o[x]$, $o[y]$ and $o[z]$ on the first layer of nodes are deterministic variables. In this situation, the only known generally applicable conditions for a correlation to be classical are certain entropic inequalities; see~\cite[Lemma~2.14]{Fri} and~\cite[Sec.~VII.B]{CLG}. One of the problematic aspects is that the set of classical correlations is not convex. In fact, this already holds in the case of a Bell scenario: although the set of classical correlations is a convex polytope on the level of the conditional distributions~\eqref{lcausal}, convexity fails as soon as one works in the space of ``unconditional'' distributions $P(o[a_1\ldots a_n x_1\ldots x_ns])$.

So it is an important problem to understand the structure of classical correlations better. One relevant question is whether one can bound the maximally necessary sizes of the hidden variable spaces $(X_e,\Sigma_e)$ in terms of the number of outcomes of each $o[v]$. In Bell scenarios, this was shown to be possible by Fine~\cite{Fine}, and it again has to do with the fact that the associated conditional distributions form a convex polytope. But in general, it is not even clear whether one must allow the hidden variable spaces $(X_e,\Sigma_e)$ to be continuous or even \emphalt{bigger than continuous} in order to obtain all classical correlations. What we have been able to prove is that finite hidden variable spaces are sufficient at least in an approximate sense:

\begin{thm}
\label{finitethm}
If $P(o[V])$ is a classical correlation on an arbitrary causal structure $G$ and $\eps > 0$, then there is another classical correlation $P'(o[V])$ on $G$ with $|P(o[V]) - P'(o[V])| < \eps$ for any value of $o[V]$ and such that the probability spaces $(X'_e,\Sigma'_e)$ occurring in the hidden variable model for $P'$ are all finite.
\end{thm}

It remains open whether this also holds without the approximation, i.e.~whether any classical correlation can be exactly represented in terms of finite hidden variable spaces. We suspect that this is not the case, but we do not have any examples. If one generalizes our definitions to the case of countably infinite $G$, as we do informally in the discussion opening Section~\ref{hbn}, then there are well-known examples~\cite[Sec.~3.6]{Upper}.

\begin{proof}
We construct, for every $\delta>0$, a classical correlation $P'_\delta$ realized by hidden variables with only finitely many values such that $|P(o[V]) - P'_\delta(o[V])|$ converges to $0$ as $\delta\to 0$. We perform this construction by induction on $|V|$, the number of nodes of the given graph.

For $|V|=1$, we have no edges and therefore no hidden variables. Hence we can simply put $P'_\delta(o[V])\defin P(o[V])$ for any $\delta$, and there is nothing to prove.

For the induction step, we start with a graph $G=(V,E)$ and a classical correlation $P(o[V])$ realized by hidden variables $\lambda_e$ as in the classical correlation equation~\eqref{ccorr}, where the hidden variable spaces $(X_e,\Sigma_e)$ may be of arbitrary size. We then choose a node $w\in V$ with $\out{w}=\emptyset$; such a node exists by the assumptions of finiteness and acyclicity. Taking out the node $w$ together with all the edges with target $w$ gives a smaller graph to which the induction hypothesis can be applied, but we need to be careful to do this in the right way. Before doing so, we need a bit more preparation.

We first consider the conditional distribution $P(o[w]|\lambda_{\inc{w}})$ at the chosen node $w$. We choose a finite subset of $[0,1]$ such that any number in $[0,1]$ has distance at most $\delta$ from this finite subset. Then we can replace the given conditional distribution $P(o[w]|\lambda_{\inc{w}})$ by a perturbation $\hat{P}_\delta(o[w]|\lambda_{\inc{w}})$ which only takes values in this finite subset and such that
\beq
\label{hatdiff}
| \hat{P}_\delta(o[w] | \lambda_{\inc{w}}) - P(o[w] | \lambda_{\inc{w}}) | < \delta \qquad \forall\, o[w],\lambda_{\inc{w}}.
\eeq
Then, for any given tuple of values for $\lambda_{\inc{w}}$, there are only finitely many possibilities for the conditional distribution $\hat{P}_\delta(o[w]|\lambda_{\inc{w}})$. Hence we can consider the family of these distributions, as indexed by $\lambda_{\inc{w}}$, as a measurable function from $X_{\inc{w}} = \prod_{e\in\inc{w}}X_e$ to some finite set $S$. We consider the product measurable space $X_\inc{w}$ to be equipped with the corresponding marginal of the total measure
\beq
\label{totaldist}
P(o[V]\lambda_E) = \prod_{v\in V} P(o[v] \lambda_{\out{v}} | \lambda_{\inc{v}}) ,
\eeq
which determines the distribution of the $\lambda_{\inc{w}}$. By virtue of Proposition~\ref{finiteapprox}, we can replace the measurable spaces $(X_e,\Sigma_e)$ for all $e\in\inc{w}$ by finite sets $X'_e$ carrying hidden variables $\mu_e$. More concretely, the $X'_e$ are coarse-grainings of the $X_e$ arising via functions $f_e : X_e\to X'_e$ and come themselves equipped with a function
\[
P'_\delta \: :\: X'_{\inc{w}} \lra S
\]
which assigns to every tuple $\mu_{\inc{w}}\in X'_{\inc{w}}$ a distribution $P'_\delta(o[w]|\mu_{\inc{w}})$ having the property that
\beq
\label{hatprime}
P'_\delta\!\left(o[w]| f_{\out{w}}(\lambda_{\out{w}})\right) = \hat{P}_\delta(o[w]|\lambda_{\inc{w}})
\eeq
holds with probability greater than $1-\delta$ with respect to the original distribution $P(\lambda_{\inc{w}})$ as the marginal of~\eqref{totaldist}.

In order to apply the induction hypothesis, we need to modify the correlation a bit by regarding each new hidden variable $\mu_e$ for $e\in\inc{w}$ as part of the outcome $o[\src{e}]$ associated to the source node $\src{e}$. In other words, we define new outcomes:
\[
\tilde{O}_v \defin \begin{cases} O_v \times  X'_{v\to w} & \textrm{if } v\to w, \\ O_v & \textrm{if } v\not\to w, \end{cases} \qquad\qquad \tilde{o}[v] \defin \begin{cases} (o[v],\mu_{v\to w}) & \textrm{if } v\to w, \\ o[v] & \textrm{if } v\not\to w. \end{cases}
\]
Now there is an obvious classical correlation on the remainder graph $G{\setminus}\{w\}$ corresponding to a distribution of the new outcomes $\tilde{o}[V]$: at those $v$ with $v\not\to w$, we retain the previous conditional distribution $P(o[v]\lambda_{\out{v}}|\lambda_{\inc{v}})$; at those $v$ with $v\to w$, we compose the old $P(o[v]\lambda_{\out{v}}|\lambda_{\inc{v}})$ with the $f_e$ from above, which gives a new conditional distribution $P(o[v]\mu_e\lambda_{\out{v}{\setminus}\{w\}}|\lambda_{\inc{v}})$, in which the first two components are now considered to be forming the new outcome $\tilde{o}[v]$. This data defines a new classical correlation $P(\tilde{o}[V])$ on the smaller graph $G{\setminus}\{w\}$.

We now apply the induction hypothesis to \emphalt{this} classical correlation. This way, we obtain finite sets $X'_e$ carrying hidden variables $\mu_e$ for all edges $e\in E$ with $\tar{e}\neq w$, where the new variables $\mu_e$ have a priori nothing to do with the previous $\mu_{\inc{w}}$ except for sharing the same letter, and conditional probabilities $P'_\delta(\tilde{o}[v]\mu_{\out{v}}|\mu_{\inc{v}})$ such that the difference
\beq
\label{indhypo}
|P'_\delta( o[{V{\setminus}\{w\}}]\, \mu_{\inc{w}} )  - P(o[{V{\setminus}\{w\}}]\, \mu_{\inc{w}} ) |  = | P'_\delta( \tilde{o}[V{\setminus}\{w\}]) - P(\tilde{o}[V{\setminus}\{w\}]) | 
\eeq
tends to $0$ as $\delta\to 0$. 

Taking all this together, we will obtain the desired bounds on the difference
\begin{align*}
P(o[V]) & - P'_\delta(o[V]) \\
& = P(o[V]) - \sum_{\mu_{\inc{w}}} P'_\delta(o[w]|\mu_{\inc{w}}) P(o[{V{\setminus}\{w\}}]\mu_{\inc{w}}) \\
& + \sum_{\mu_{\inc{w}}} \left( P'_\delta(o[w]|\mu_{\inc{w}}) P(o[{V{\setminus}\{w\}}]\mu_{\inc{w}}) - P'_\delta(o[w]|\mu_{\inc{w}}) P'_\delta(o[{V{\setminus}\{w\}}] \mu_{\inc{w}})\right),
\end{align*}
where we have inserted the first sum such that it cancels with the first part of the second sum. The second sum can be bounded in terms of~\eqref{indhypo}: by making the $\delta$ used in the application of the induction hypothesis arbitrarily small, we can also make the second sum arbitrarily small. Thus it remains to bound the difference between $P(o[V])$ and the first sum, which can be done like this:
\begin{align*}
\bigg| P(o[V]) - & \sum_{\mu_{\inc{w}}} P'_\delta(o[w]|\mu_{\inc{w}}) P(o[{V{\setminus}\{w\}}]\mu_{\inc{w}}) \bigg| \\
& = \left| P(o[V]) - \int_{\lambda_{\inc{w}}} P'_\delta(o[w]|f_{\inc{w}}(\lambda_{\inc{w}})) P(o[{V{\setminus}\{w\}}]\lambda_{\inc{w}}) \right| \\
& \stackrel{\eqref{hatprime}}{\leq} \left| P(o[V]) - \int_{\lambda_{\inc{w}}} \hat{P}_\delta(o[w]|\lambda_{\inc{w}}) P(o[{V{\setminus}\{w\}}]\lambda_{\inc{w}}) \right| + \delta \\
& \stackrel{\eqref{hatdiff}}{\leq} \left| P(o[V]) - \int_{\lambda_{\inc{w}}} P(o[w]|\lambda_{\inc{w}}) P(o[{V{\setminus}\{w\}}]\lambda_{\inc{w}}) \right| + 2 \delta \\
& = 2\delta.
\end{align*}
\end{proof}

In Bell scenarios, it is well-known that the ``information processing'' at the measurement nodes $a_i$ may be assumed to be deterministic without loss of generality~\cite{Fine}, in the sense that any classical correlation has a hidden variable model satisfying this property. For general causal structures, it also seems very plausible that it should be possible to assume that the information processing at each node $v$ with $\pa{v}\neq\emptyset$ is deterministic. The idea here is that if $v$ has at least one parent node, then the randomness needed can be ``pushed back'' in the sense of generating the randomness at the parent node and then transporting it from there to $v$ by augmenting the hidden variable on that edge with this additional information. If one repeats this process of pushing randomness generation back to the parents nodes, then one eventually ends up with a classical model in which all information processing is deterministic, except at those nodes that do not have any parents, which is where all the randomness originates.

More precisely, the hypothesis is this:

\begin{conj}
\label{detconj}
If $P(o[V])$ is a classical correlation, then $P(o[V])$ can be represented in the form~\eqref{ccorr} in such a way that all conditional probabilities are deterministic at all nodes $v$ with $\inc{v}\neq\emptyset$, meaning that $P(o[v] \lambda_{\out{v}}|\lambda_{\inc{v}})$ is a Dirac measure on $O_v\times X_{\out{v}}$ for all values of $\lambda_{\inc{v}}$.
\end{conj}

Due to the measure-theoretical technicalities involved, we only have a partial proof of this conjecture covering the case in which all hidden variables take only finitely many values. In this case, the statement can be reformulated as saying that $P(o[v]\lambda_{\out{v}}|\lambda_{\inc{v}})\in\{0,1\}$ for any values of the variables involved. Our proof is based on exactly the intuitive idea of repeatedly pushing the randomness back to parent nodes:

\bigskip
\begin{partialproof}
As in the proof of Theorem~\ref{finitethm}, we use induction on the size of $G$. When $G$ has only vertex, then there is no hidden variable and hence nothing to prove.

For the induction step, we start with a graph $G=(V,E)$ and a classical correlation $P(o[V])$ on $G$ realized by hidden variables $\lambda_e$ living on finite sets $X_e$ and conditional distributions $P(o[v]\lambda_{\out{v}}|\lambda_{\inc{v}})$ which are, due to the finiteness assumption, simply stochastic matrices. If there is no node with a parent, then the graph has no edges at all, and again there is nothing to prove. Otherwise, we pick a vertex $w$ with $\pa{w}\neq\emptyset$, but $\ch{w}=\emptyset$; since $G$ is finite and acyclic, there is at least one such vertex which can be found e.g.~by starting at any node with an incoming edge and following outgoing edges from node to node until one reaches a node without any outgoing edges. From now on, we keep such a node $w$ fixed.

The induction step consists in ``pushing back'' the randomness inherent in the conditional distribution $P(o[w]|\lambda_{\inc{w}})$ to a parent node of $w$, and applying the induction hypothesis afterwards. While this should be intuitively plausible, we will see that the technical details are quite demanding.

So to perform this induction step, pick an arbitrary parent $u\in\pa{w}$, which is where the randomness will be pushed back to. Doing so means that the hidden variable space $X_{u\to w}$ will have to be enlarged, and the appropriate definition turns out to be
\[
X'_{u\to w} \defin O_w^{X_{\inc{w}}}\times X_{u\to w},
\]
where the first factor is the set of all functions $X_{\inc{w}}\to O_w$. The given conditional distribution $P(o[w]|\lambda_{\inc{w}})$ can be regarded as a distribution on this function space by assigning to a function $f:X_\inc{w}\to O_w$ the probability
\[
\hat{P}(f) \defin \prod_{\lambda_{\inc{w}}\in X_{\inc{w}}} P(f(\lambda_{\inc{w}})|\lambda_{\inc{w}}).
\]
Intuitively, a function $f:X_{\inc{w}}\to O_w$ can also be regarded as a list of elements of $O_w$ indexed by $X_{\inc{w}}$, and then this definition says that all the elements of this list are selected independently of each other according to their respective distributions determined by the given $P(o[w]|\lambda_{\inc{w}})$.

If we denote the new hidden variable with the new domain $X'_{u\to w}$ by $\lambda'_{u\to w}$, then we can define a new conditional distribution $P'(o[w]|\lambda'_{u\to w}\lambda_{\inc{w}{\setminus}\{u\to w\}})$ by writing $\lambda'_{u\to w}$ as a pair $\lambda'_{u\to w}=(f,\lambda_{u\to w})$ and putting
\[
P'(o[w]|(f,\lambda_{u\to w})\lambda_{\inc{w}{\setminus}\{u\to w\}}) \defin \begin{cases} 1 & \textrm{if } f(\lambda_{\inc{w}}) = o[w], \\ 0 & \textrm{otherwise},\end{cases}
\]
where $\lambda_{\inc{w}}$ stands for the tuple of hidden variables consisting of the $\lambda_{\inc{w}{\setminus}\{u\to w\}}$ together with the additional $\lambda_{u\to w}$ which is part of $\lambda'_{u\to w}$. By construction, $P'$ is deterministic information processing: the node $w$ receives a collection of incoming hidden variables $\lambda_{\inc{w}}$ as well as a function $f:X_{\inc{w}}\to O_w$ and does nothing but to apply the latter to the former in a deterministic manner.

In order to define a classical hidden variable model using the new hidden variable space $X'_{u\to w}$, we also need to modify the conditional distribution $P(o[u]\lambda_{\out{u}}|\lambda_{\inc{u}})$ to some $P'(o[u]\lambda'_{u\to w}\lambda_{\out{u}{\setminus}\{u\to w\}}|\lambda_{\inc{u}})$. Since $\lambda'_{u\to w} = (f,\lambda_{u\to w})$, we can do this by simply putting
\[
P'(o[u] f \lambda_{\out{u}}|\lambda_{\inc{u}}) \defin \hat{P}(f) \cdot P(o[u] \lambda_{\out{u}} | \lambda_{\inc{u}}),
\]
where $\hat{P}(f)$ is as above. In other words, the function $f$ is chosen already at $u$, randomly and independently of any of the other hidden variables. In this sense, the randomness has been ``pushed back'' from the node $w$ to the node $u$.

In order to show that the two new conditional distributions $P'$ at nodes $u$ and $w$ indeed give rise to the same overall classical correlation, it is enough to show that the resulting conditional distribution $P'(o[uw] \lambda_{\out{u}{\setminus}\{u\to w\}}|\lambda_{\inc{u}\cup\inc{w}{\setminus}\{u\to w\}})$ coincides with the original one, since this is the only part of the graph to which we have applied any modifications. Indeed,
\begin{align*}
P'(o[u w] & \lambda_{\out{u}{\setminus}\{u\to w\}}|\lambda_{\inc{u}\cup\inc{w}{\setminus}\{u\to w\}}) \\
& = \sum_{f,\lambda_{u\to w}} P'(o[w] | f \lambda_{\inc{w}}) \, P'(o[u] f \lambda_{\out{u}} | \lambda_{\inc{u}}) \\
& = \sum_{f,\lambda_{u\to w}} \delta_{f(\lambda_{\inc{w}}),o[w]} \cdot \hat{P}(f)\cdot P(o[u]\lambda_{\out{u}} | \lambda_{\inc{u}}) \\
& = \sum_{\lambda_{u\to w}} P(o[u]\lambda_{\out{u}}|\lambda_{\inc{u}}) \quad\sum_f \delta_{f(\lambda_{\inc{w}}),o[w]} \prod_{\mu_{\inc{w}}} P(f(\mu_{\inc{w}})|\mu_{\inc{w}}) \\
& \stackrel{(*)}{=} \sum_{\lambda_{u\to w}} P(o[u]\lambda_{\out{u}}|\lambda_{\inc{u}}) \:\sum_{f(\lambda_{\inc{w}})} \delta_{f(\lambda_{\inc{w}}),o[w]}\, P(f(\lambda_{\inc{w}})|\lambda_{\inc{w}}) \\
& = \sum_{\lambda_{u\to w}} P(o[u]\lambda_{\out{u}}|\lambda_{\inc{u}}) \cdot P(o[w]|\lambda_{\inc{w}}) \\
& = P(o[uw]\lambda_{\out{u}{\setminus}\{u\to w\}}|\lambda_{\inc{u}\cup\inc{w}{\setminus}\{u\to w\}}),
\end{align*}
where the step $(*)$ uses the fact that upon regarding $f$ as a list of elements of $O_w$ indexed by $X_{\inc{w}}$, all components of this list except for $f(\lambda_{\inc{w}})$ simply drop out by normalization of probability.

This finishes our construction of making the information processing at $w$ deterministic by pushing back the randomness needed $u$ along the edge $u\to w$. The induction step can now be completed as in the proof of Theorem~\ref{finitethm}: the hidden variables $\lambda_{\inc{w}}$ have to be regarded as part of the outcomes $o[v]$ for those nodes $v$ with $v\to w$, so that we obtain an induced classical correlation on the restricted graph $G{\setminus}\{w\}$. Applying the induction hypothesis to this situation then results in a classical model for the corresponding correlation $P(\tilde{o}[V{\setminus}\{w\}])=P(o[{V{\setminus}\{w\}}] \lambda_{\inc{w}})$ with the property that all information processing at the non-root nodes is deterministic. Putting this together with the deterministic processing $P'$ at node $w$ completes the induction step.
\end{partialproof}

By virtue of Theorem~\ref{finitethm}, this also shows that any classical correlation can be arbitrarily well approximated by one in which all conditional distributions $P(o[v]\lambda_{\out{v}}|\lambda_{\inc{v}})$ are deterministic. We suspect that Conjecture~\ref{detconj} also holds without the approximation, but we do not have a proof of this.\footnote{Our partial proof could possibly be extended to a full proof if the category of measurable spaces was cartesian closed, but this seems to be an open problem. See \href{http://mathoverflow.net/questions/1388/is-there-a-natural-measures-on-the-space-of-measurable-functions}{http://mathoverflow.net/questions/1388/is-there-a-natural-measures-on-the-space-of-measurable-functions}.} The upcoming Theorem~\ref{hbnthm} might be helpful for achieving a complete proof.

\newpage
\section{Correlations within any theoretical framework}
\label{seccats}

It should be useful to have a definition of correlation not just for the classical case, but also for other mathematical frameworks for physical theories. In particular, there should be definitions of quantum correlation and of general probabilistic correlations. We consider this problem by outlining the form in which a framework for physical theories, like ``classical'', ``quantum'' or ``general probabilistic'',  needs to be given in order for it to enable a corresponding definition of correlation. Our answer to this turns out to be very similar to the definition of \emph{causal operational-probabilistic theory} of Chiribella, D'Ariano and Perinotti~\cite{CDP}. The reason for tackling this very general problem at this stage is that it seems to give the most convenient rigorous definition of quantum correlation. Those readers who are not interested in mathematically precise definitions can directly proceed to the informal definition of quantum correlations in Section~\ref{secqc}. Considerations similar to the ones of this section can also be found in the independent work of Henson, Lal and Pusey~\cite{HLP}.

So what pieces of structure did we need in the previous section to define classical correlations? For one thing, each edge of the graph $G=(V,E)$ carries a classical system defined in terms of a measurable space. Second, the nodes $v\in V$ in turn are implemented as ``information processing gates'' given by stochastic instruments indexed by outcomes $o[v]$; here, an individual operation is typically subnormalized, i.e.~happens with a state-dependent success probability. Only when taking the sum over all $o[v]$ do we get a normalization-preserving operation.

In general, the dichotomy between labelling the edges by state spaces and the nodes by operations suggests that we should assume that a framework for physical theories should be given by a \emph{category} which we denote $\Cp$.\footnote{We refer to~\cite{CP,Spivak} for background on the relevant pieces of category theory.} Then, each edge gets labelled by an object in $\Cp$, while each node gets labelled by a morphism of $\Cp$. More precisely, we assume that $\Cp$ is a (strict) \emph{symmetric monoidal category} whose morphisms play the role of the operations; in particular, this means that no assumption like normalization or preservation of probability should be made when defining the morphisms of $\Cp$. In order to be able to make sense of the requirement that the every node carries one operation \emphalt{for each outcome} and, when summing over all outcomes, the resulting total operation is normalized (more accurately, normalization-preserving), we need to make several additional assumptions:

\renewcommand{\labelenumi}{(cat.\roman{enumi})}
\renewcommand{\theenumi}{(cat.\roman{enumi})}

\begin{enumerate}
\itemsep6pt
\item\label{Ccone} Every morphism set $\Cp(X,Y)$ is a convex cone: any morphism in $\Cp(X,Y)$ can be multiplied by a positive scalar, any two morphisms $\Cp(X,Y)$ can be added, and there is a neutral element $0\in\Cp(X,Y)$. More precisely, $\Cp(X,Y)$ is assumed to be a module over the semiring $\Rplus$. Intuitively, the consequences of this are that we can take stochastic mixtures of operations, and the normalization of a morphism can be rescaled by simply multiplying the morphism by a scalar.
\item\label{Clin} The composition of morphisms
\[
\Cp(X,Y) \times \Cp(Y,Z) \to \Cp(X,Z)
\]
is bilinear with respect to this $\Rplus$-module structure. In other words, composition distributes both over addition and scalar multiplication of morphisms. Likewise, the tensor product
\[
\Cp(X_1,Y_1) \times \Cp(X_2,Y_2) \to \Cp(X_1\otimes X_2,Y_1\otimes Y_2)
\]
is also bilinear.
\item\label{CProcess} In addition, there is a distinguished symmetric monoidal subcategory $\Pp\subseteq\Cp$ containing all identity morphisms. We think of $\Pp$ as the subcategory which contains all the normalized operations.
\item\label{CI} We have $\Cp(I,I)=\Rplus$ and $\Pp(I,I)=\{1\}$, where $I$ is the monoidal unit of $\Cp$ (and $\Pp$). In other words, there is exactly one normalized operation which does nothing to nothing, and this operation is denoted $1$. The operations which act on nothing are precisely the scalar multiples of $1$. This how we are going to obtain probabilities: as morphisms of type $I\to I$ in $\Cp$, which are nonnegative real numbers.
\item\label{CPterm} We have $\Pp(X,I) = \{1\}$, i.e.~the unit object $I$ is terminal in $\Pp$. In other words, every system $X$ can be discarded in a unique way in the sense that there is exactly one normalized operation from $X$ to no system which we denote $\tau_X:X\to I$ and draw as ``$\ground$''. This assumption corresponds to the causality assumption of~\cite{CDP}, which also forms a central aspect of the ``causal categories'' of~\cite{CLcausal}, where the use of the ground symbol ``$\ground$'' for $\tau$ has been introduced.
\end{enumerate}
\vspace{6pt}

This finishes the list of requirements that the category $\Cp$ should satisfy in order for it to be possible to define the notion of ``$\Cp$-correlation''. For example, in the quantum case that we will consider in the next section, $\Cp$ will be the category of completely positive maps in which Hilbert spaces are the objects and completely positive maps between their sets of states will be the morphisms. As a more basic example, we can also define the classical correlations from the previous section like this by defining a category in which the objects are measurable spaces and the morphisms are stochastic operations. We will do this in more detail below.

Once a category $\Cp$ together with $\Pp\subseteq\Cp$ as above has been determined, how do we obtain the definition of $\Cp$-correlation on a graph $G$? The basic idea is to interpret $G$ as a diagram within the category $\Cp$, assuming that $G$ has been labelled by an object $X_e\in\Cp$ for every edge $e\in E$, and by a morphism 
\[
g_v \:\in\: \Cp\left(\bigotimes_{e\in\inc{v}} X_e ,\bigotimes_{e\in\out{v}} X_e\right) 
\]
for every node $v\in V$. Labelling the edges by objects and the nodes by morphisms is essentially the same idea as the one which underlies ``quantum picturalism''~\cite{Coecke}, namely the graphical calculus for symmetric monoidal categories~\cite{JS}. This is another reason for why we believe that taking the hidden variables to live on the edges---as opposed to the nodes---is the appropriate thing to do. Once $G$ can be interpreted as a diagram like this, it evaluates to an overall composite morphism $I\to I$, which~\ref{CI} guarantees to be a nonnegative number. We interpret this number as the probability that the specific labelled graph $G$ ``happens'', exactly as in Hardy's operator-tensor formulation~\cite{HardyOT}.

More precisely, we want to label each node $v$ not just by a single operation, but by a whole collection of operations indexed by the outcome $o[v]$. In other words, we need to assume that $v$ is labelled by a function from outcomes to morphisms,
\beq
\label{Cinstrument}
f_v \: :\: O_v \longrightarrow \Cp\left(\bigotimes_{e\in\inc{v}} X_e ,\bigotimes_{e\in\out{v}} X_e\right) .
\eeq
This function is assumed to assign to every outcome $o[v]$ a $\Cp$-operation $f_v(o[v]) : \bigotimes_{e\in\inc{v}} X_e \to \bigotimes_{e\in\out{v}} X_e$ in such a way that $\sum_{o[v]} f_v(o[v])$ is normalized, i.e.~lies in the subcategory $\Pp$. We can call this $f_v$ a \emph{$\Cp$-instrument} in analogy with the notion of quantum instrument, since it is a collection of operations indexed by an outcome such that the total operation obtained by marginalization (summing) over the outcome is normalized. The $\Cp$-instrument $f_v$ plays the role of the ``information processing gate'' turning all the incoming systems into outgoing systems and producing a classical outcome while doing so. We think of a particular operation $f_v(o[v])$ as realized at node $v$ exactly when the particular outcome $o[v]$ occurs. So when computing the probability $P(o[V])$ of a specific joint outcome $o[V]$ to occur, every node $v\in V$ is therefore labelled by exactly one morphism $f_v(o[v])$. Then one can try to interpret the graph $G$, labelled by the objects $X_e$ and these particular morphisms $f_v(o[v])$, as a diagram in $\Cp$, which evaluates to an element of $\Cp(I,I)$. Indeed, the graphical calculus for symmetric monoidal categories, as formalized by Joyal and Street~\cite[Sec.~2.1]{JS} as the definition of ``value'' of a progressive polarized diagram, should be the appropriate way of doing this. By~\ref{CI}, this value $I\to I$ is a nonnegative real number, which we take to be the definition of $P(o[V])$. 

However, there is an annoying technical problem with the idea of interpreting the graph $G$, when appropriately labelled, as a diagram in $\Cp$. Strictly speaking, the two tensor products $\bigotimes_{e\in\inc{v}} X_e$ and $\bigotimes_{e\in\out{v}} X_e$ are not defined in $\Cp$: as a symmetric monoidal category, we can only make sense of binary tensor products in $\Cp$.\footnote{There is an ``unbiased'' definition of monoidal category in which tensor products are not necessarily binary~\cite[Sec.~3.1]{Leinster}, but an ordering of the factors does still need to be chosen. More generally, there is also a definition of ``fat symmetric multicategory''~\cite[App.~A]{Leinster}, which is more general than the definition of ``order-independent unbiased symmetric monoidal category'' that we would need. However, there does not yet seem to be any theorem on interpreting directed acyclic graphs as diagrams in symmetric monoidal categories, which is the reason for the technical complications that we are facing here. See also \href{http://nforum.mathforge.org/discussion/3101/symmetric-monoidal-category/?Focus=46303\#Comment_46303}{nforum.mathforge.org/discussion/3101/symmetric-monoidal-category/?Focus=46303{\#}Comment\_46303}.} Hence, following Joyal and Street~\cite{JS}, we equip each set $\inc{v}$ and $\out{v}$ with a particular \emphalt{linear ordering}. We will show later that the resulting set of $\Cp$-correlations does not depend on this choice of ordering. Having fixed such orderings, we can now make sense of the tensor products $\bigotimes_{e\in\inc{v}} X_e$ and $\bigotimes_{e\in\out{v}} X_e$ by defining them as iterated binary tensor products with respect to the given ordering.

In order to describe the notion of value in more detail, we need to introduce some additional terminology. The details of this will differ from the treatment of Joyal and Street, but it should be clear to the interested reader how to translate their definitions into ours. The general notion of value applies to directed acyclic graphs which may have incoming and outgoing ``half-edges'', i.e.~edges which have only a source or target node within the graph. We can formalize this by defining a \emph{directed acyclic graph with boundary} $G$ to consist of a set $V$ of nodes, a set $E$ of ordinary edges, a set $E_{\mathrm{in}}$ of incoming edges, and a set $E_{\mathrm{out}}$ of outgoing edges, together with source and target maps
\[
s : E\cup E_{\mathrm{out}} \to V,\qquad t : E\cup E_{\mathrm{out}} \to V.
\]
So if we put $\inc{v}\defin t^{-1}(v)$, then this set of incoming edges at $v$ also includes all half-edges that have $v$ as their target, and similarly for $\out{v}$. The acyclicity condition is then as before in Definition~\ref{defdag}, and $G$ is acyclic if and only if the restricted graph $(V,E)$ without the half-edges is so. During the discussion of value, we use the term ``graph'' as a synonym of ``directed acyclic graph with boundary''. As a special case, the boundary may be empty, in which case the definition reduces to the usual notion of directed acyclic graph which we have been using until now.

Such a graph is \emphalt{polarized} if each set $\inc{v}$ and $\out{v}$ is equipped with a linear ordering. It is \emphalt{anchored} if both $E_{\mathrm{in}}$ and $E_{\mathrm{out}}$ are equipped with linear orderings. A \emphalt{valuation} in $\Cp$ is a labelling of $E$, $E_{\mathrm{in}}$ and $E_{\mathrm{out}}$ by objects $X_e\in\Cp$ and of $V$ by morphisms $g_v$ of $\Cp$ such that each $f_v$ is of type
\[
g_v \: : \: \bigotimes_{e\in\inc{v}} X_e \longrightarrow \bigotimes_{e\in\out{v}} X_e,
\]
where it is understood that the tensor products are taken with respect to the chosen orderings.

Then Joyal and Street~\cite[Sec.~2.1]{JS} show that it is possible to assign a value to this graph in the sense of regarding it as a diagram in $\Cp$ and composing it to a morphism 
\[
\nu(G) : \bigotimes_{e\in E_{\mathrm{in}}} X_e\longrightarrow \bigotimes_{e\in E_{\mathrm{out}}} X_e,
\]
where again these tensor products are with respect to the chosen orderings on $E_{\mathrm{in}}$ and $E_{\mathrm{out}}$. If $E_{\mathrm{in}}=E_{\mathrm{out}}=\emptyset$, as in the case of interest to us, then $\nu(G) : I\to I$. In fact, it can be shown that there is a unique way of assigning a value satisfying the following properties: 

\renewcommand{\labelenumi}{(val.\roman{enumi})}
\renewcommand{\theenumi}{(val.\roman{enumi})}

\begin{enumerate}
\item\label{velem} If $G$ is a graph containing exactly one node, then $\nu(G)$ coincides with the morphism labelling that node.
\item\label{vmult}
If $G_1$ and $G_2$ are graphs equipped with valuations in $\Cp$, then also the disjoint union graph $G_1 + G_2$ carries a valuation in $\Cp$, and
\[
\label{numult}
\nu(G_1 + G_2) = \nu(G_1) \otimes \nu(G_2) .
\]
Intuitively, this corresponds to placing the two diagrams $G_1$ and $G_2$ side by side, and hence their value should be interpreted as the tensor product of the individual values.
\item\label{vcomp}
The second requirement is similar and refers to composition of graphs: if $G_1$ and $G_2$ are such that there is an equality of ordered sets $E_{1,\mathrm{out}} = E_{2,\mathrm{in}}$, and both $G_1$ and $G_2$ carry a valuation such that the object assigned to a half-edge $e\in E_{1,\mathrm{out}}$ coincides with the object associated to the corresponding half-edge $e\in E_{2,\mathrm{in}}$, then the two graphs can be composed to $G_2 \circ G_1$ in the obvious way, and this composite graph also carries a valuation. The requirement now is that
\[
\label{nucomp}
\nu(G_2 \circ G_1) = \nu(G_2) \circ \nu(G_1),
\]
where ``$\circ$'' on the right-hand side is composition in $\Cp$.
\item\label{vlin} Taking the value is a linear function of the morphism $g_v$ at each node $v\in V$.
\end{enumerate}

\renewcommand{\labelenumi}{(\alph{enumi})}
\renewcommand{\theenumi}{(\alph{enumi})}

Armed with these preliminary considerations, we can now get to our main result of this section. If we have a graph without boundary $G$, and $G$ is labelled by objects on the edges and $\Cp$-instruments~\eqref{Cinstrument} on the nodes, then every joint outcome $o[V]$ can be assigned a number $P(o[V])$ by evaluating it as a diagram in $\Cp$ as described above. The main point is this:

\begin{thm}
\label{Ccorriscorr}
The resulting numbers $P(o[V])$ form a correlation on $G$.
\end{thm}

See also~\cite{Cns} for more general considerations relating assumption~\ref{CPterm} to the no-signaling property.

\begin{proof}
This proof generalizes the one of Proposition~\ref{ccorriscorr}. In particular, the overall proof strategy is the same, and we will see that Proposition~\ref{ccorriscorr} can indeed be regarded as an instance of the present claim. 

\begin{figure}
\begin{tikzpicture}[node distance=2.1cm,>=stealth',thick,circuit ee IEC]
\tikzstyle{transition}=[rectangle,thick,draw=black!75,fill=black!10,minimum size=6.0mm]
\node[transition] (cent) at (0,2) {$g$} ;
\node (in) at (0,0) {} ;
\node[ground,rotate=90] (out1) at (-1.5,4) {} ;
\node (out2) at (1.5,4) {} ;
\draw[->] (in) -- (cent) node [midway,right] {$X$} ;
\draw[->] (cent) -- (out1) node [midway, below left] {$Y$} ;
\draw[->] (cent) -- (out2) node [midway, below right] {$Z$} ;
\end{tikzpicture}
\caption[]{Taking the marginal of a morphism in $\Cp$ over an output system $Y$ is defined in terms of composition with the unique morphism $\tau_Y\in\Pp(Y,I)$ drawn as ``$\ground$''. In the quantum case, this is an application of partial trace over $Y$.}
\label{margmorph}
\end{figure}
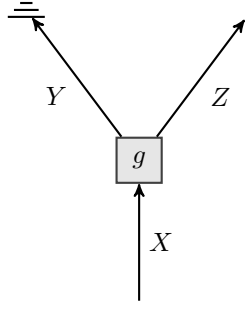

We start by deriving the following generalization of~\eqref{cmarginal}: if $W\subseteq V$ is a set of nodes with $\pst{W}=W$, then the marginal $P(o[W])$ can be computed using the following steps:
\begin{itemize}
\item Take the restricted graph $(W,E|_W)$, where $E|_W$ contains exactly those edges whose source and target both lie in $W$. Just as the original $G$, this is also a graph without boundary.
\item Equip $(W,E|_W)$ with the same objects $X_e$ as edge labels, and label the nodes by the $\Cp$-instruments whose components are the corresponding marginal operations (Figure~\ref{margmorph}). Here, the marginal of an operation $f_v(o[v])$ needs to be taken over all those edges $e\in\out{v}$ whose target $\tar{e}$ does not lie in $W$. It follows from bilinearity~\ref{Clin} and the fact that the marginalization morphisms $\tau_{\tar{e}}$ are in the subcategory $\Pp$ that the new marginal operations do also form a $\Cp$-instrument.
\item Take the linear orderings on the new sets $\inc{v}$ and $\out{v}$ for all $v\in W$ to be those induced from $G$.
\item Evaluate this graph to a family of numbers $P(o[W])$ by applying $\nu$ as described above.
\end{itemize}

We prove that the resulting numbers $P(o[W])$ coincide with the corresponding marginal of $P(o[V])$ by induction on the size of $V{\setminus} W$. The base case is $W=V$, for which the claim is simply the definition of $P(o[V])$. For the induction step, we need to prove that the claim holds for a given $W$ with $\pst{W}=W$, while assuming that it holds for all bigger $W'$ in place of $W$ which likewise satisfy $\pst{W'}=W'$. For the given $W$, we can choose a node $w$ with $\pa{w}\subseteq W$, as in the proof of~\eqref{cmarginal}. Upon putting $W'\defin W\cup\{w\}$, it is then guaranteed that $\pst{W'}=W'$. We therefore can apply the induction assumption, which gives that the marginal $P(o[W'])$ can be computed by applying the above steps to $W'$; this is depicted in Figure~\ref{before}. Then it remains to evaluate the sum over $o[w]$, in order to show that the further restricted marginal
\beq
\label{margclaim}
P(o[W]) = \sum_{o[w]} P(o[{W'}])
\eeq
can likewise be computed by following the above prescription. In order to prove this, we use the linearity property~\ref{vlin} to interpret the right-hand side as the value of the diagram arising from the valuation of $(W',E|_{W'})$ by objects and operations such that each $v\in W$ is labelled by the morphism $f_v(o[v])$ where the value $o[v]$ is the given one on the left-hand side of~\eqref{margclaim}, while the ``new'' node $w$ is labelled by $\sum_{o[w]} f_w(o[w])$. Because the latter is a morphism in $\Pp$, we can apply assumption~\ref{CPterm}: since $I$ is a terminal object in $\Pp$, the set $\Pp\big(\bigotimes_{e\in \inc{w}} X_e, I\big)$ has exactly one element. In particular, the elements $\sum_{o[w]} f_w(o[w])\in\Pp\big(\bigotimes_{e\in \inc{w}} X_e,I\big)$ and $\bigotimes_{e\in\inc{w}} \tau_{X_e}\in \Pp\big(\bigotimes_{e\in \inc{w}} X_e, I\big)$ are equal,
\beq
\label{tmor}
\sum_{o[w]} f_w(o[w]) = \bigotimes_{e\in\inc{w}} \tau_{X_e},
\eeq
where $\tau_{X_e}$ is itself the unique element of $\Pp(X_e,I)$. Intuitively, the idea is that $\sum_{o[w]} f_w(o[w])$ is the unique normalized operation going from the composite system $\bigotimes_{e\in\inc{w}} X_e$ to no system, as illustrated in Figure~\ref{first}, and hence coincides with throwing away all systems in parallel, which is given by the morphism $\bigotimes_{e\in\inc{w}} \tau_{X_e}$ as depicted in Figure~\ref{intermediate}. Now we can regard the graph $(W',E|_{W'})$ as the composition of two parts as indicated by the dashed line in Figure~\ref{first}. The upper part consists of only one node with ingoing edges and is labelled by the morphism~\eqref{tmor}, and hence by~\ref{velem} its value coincides with this very morphism. Since this morphism itself factors into a tensor product, rule~\ref{vmult} tells us that we obtain the same value if we furnish each incoming edge with its own target node, as in Figure~\ref{intermediate}. Then again by the rules~\ref{vmult} and~\ref{vcomp}, the overall diagram value coincides with the one obtained by taking $(W,E|_W)$ and labelling it with marginal morphisms as described above. This completes the induction step.

\begin{figure}
\centerfloat
\subfigure[]{
\label{before}
\begin{tikzpicture}[node distance=2.1cm,>=stealth',thick,scale=.7]
\tikzstyle{transition}=[rectangle,thick,draw=black!75,fill=black!10,minimum size=6.0mm]
\node[transition] (aw) at (0,5) {$f_w(o[w])$} ;
\node[transition] (in1) at (-2,3) {$f_u(o[u])$} ;
\node[transition] (in2) at (2,3) {$f_v(o[v])$} ;
\draw[<-] (aw) -- (in1) ;
\draw[<-] (aw) -- (in2) ;
\node (e1) at (-3.5,1) {} ;
\node (e2) at (-0.5,1) {} ;
\node (e3) at (.5,1) {} ;
\node (e4) at (2,1) {} ;
\node (e5) at (3.5,1) {} ;
\node[rotate=35] (x) at (-3.5,5) {$\vdots$} ;
\draw[->] (in1) -- (x) ;
\draw[->] (e1) -- (in1) ;
\draw[->] (e2) -- (in1) ;
\draw[->] (e3) -- (in2) ;
\draw[->] (e4) -- (in2) ;
\draw[->] (e5) -- (in2) ;
\node at (-3.6,.7) {$\vdots$} ;
\node at (-1.8,.7) {$\vdots$} ;
\node at (0,.7) {$\vdots$} ;
\node at (1.8,.7) {$\vdots$} ;
\node at (3.6,.7) {$\vdots$} ;
\end{tikzpicture}}\hspace{2pc}
\subfigure[]{
\label{first}
\begin{tikzpicture}[node distance=2.1cm,>=stealth',thick,scale=.7,circuit ee IEC]
\tikzstyle{transition}=[rectangle,thick,draw=black!75,fill=black!10,minimum size=6.0mm]
\node[ground,rotate=90] (aw) at (0,5) {} ;
\node[transition] (in1) at (-2,3) {$f_u(o[u])$} ;
\node[transition] (in2) at (2,3) {$f_v(o[v])$} ;
\draw[<-] (aw) -- (in1) ;
\draw[<-] (aw) -- (in2) ;
\node (e1) at (-3.5,1) {} ;
\node (e2) at (-0.5,1) {} ;
\node (e3) at (.5,1) {} ;
\node (e4) at (2,1) {} ;
\node (e5) at (3.5,1) {} ;
\node[rotate=35] (x) at (-3.5,5) {$\vdots$} ;
\draw[->] (in1) -- (x) ;
\draw[->] (e1) -- (in1) ;
\draw[->] (e2) -- (in1) ;
\draw[->] (e3) -- (in2) ;
\draw[->] (e4) -- (in2) ;
\draw[->] (e5) -- (in2) ;
\node at (-3.6,.7) {$\vdots$} ;
\node at (-1.8,.7) {$\vdots$} ;
\node at (0,.7) {$\vdots$} ;
\node at (1.8,.7) {$\vdots$} ;
\node at (3.6,.7) {$\vdots$} ;
\draw[dashed] plot [smooth] coordinates {(-2.2,5) (-.35,3.8) (.65,3.8) (2.2,4.5)} ;
\end{tikzpicture}}\hspace{2pc}
\subfigure[]{
\label{intermediate}
\begin{tikzpicture}[node distance=2.1cm,>=stealth',thick,scale=.7,circuit ee IEC]
\tikzstyle{transition}=[rectangle,thick,draw=black!75,fill=black!10,minimum size=6.0mm]
\node[ground,rotate=90] (w1) at (-1,5) {} ;
\node[ground,rotate=90] (w2) at (1,5) {} ;
\node[transition] (in1) at (-2,3) {$f_u(o[u])$} ;
\node[transition] (in2) at (2,3) {$f_v(o[v])$} ;
\draw[->] (in1) -- (w1) ;
\draw[->] (in2) -- (w2) ;
\node (e1) at (-3.5,1) {} ;
\node (e2) at (-0.5,1) {} ;
\node (e3) at (.5,1) {} ;
\node (e4) at (2,1) {} ;
\node (e5) at (3.5,1) {} ;
\node[rotate=35] (x) at (-3.5,5) {$\vdots$} ;
\draw[->] (in1) -- (x) ;
\draw[->] (e1) -- (in1) ;
\draw[->] (e2) -- (in1) ;
\draw[->] (e3) -- (in2) ;
\draw[->] (e4) -- (in2) ;
\draw[->] (e5) -- (in2) ;
\node at (-3.6,.7) {$\vdots$} ;
\node at (-1.8,.7) {$\vdots$} ;
\node at (0,.7) {$\vdots$} ;
\node at (1.8,.7) {$\vdots$} ;
\node at (3.6,.7) {$\vdots$} ;
\draw[dashed] plot [smooth] coordinates {(-2.2,5) (-.35,3.8) (.65,3.8) (2.2,4.5)} ;
\draw[dashed] (.1,3.7) -- (.1,5.3) ;
\end{tikzpicture}}
\caption[]{Illustration of the induction step in the proof of Theorem~\ref{Ccorriscorr} for part of an example graph. The ground symbol $\ground$ represents the unique normalized operation $\tau_X\in\Pp(X,I)$ from any system $X$ to no system.}
\label{proofillu}
\end{figure}
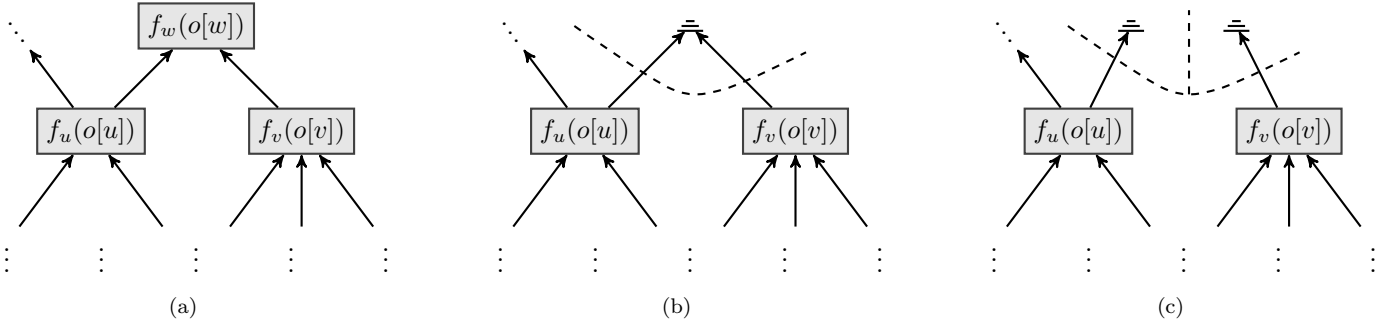

Now we can make use of the above prescription. First, if $W=\emptyset$, then we obtain the marginal $P(o[{\emptyset}])$, which simply is a number corresponding to the normalization of the distribution $P(o[V])$, which should be $1$. What the above prescription tells is that we can compute $P(o[{\emptyset}])$ as the value of the empty graph, which is $\id_I\in\Cp(I,I)$. As a number, this is indeed $1$.

Second, we can also show that the ``binary'' correlation equation~\eqref{bicorrelation} with the assumption $\pst{U}=U$ and $\pst{W}=W$ holds by applying the above prescription to the subgraph $(U\cup W,E|_{U\cup W})$. Since this subgraph decomposes into the disjoint union of $(U,E|_U)$ and $(W,E|_W)$, the claim now follows immediately from~\ref{vmult}.
\end{proof}

So for any category of operations $\Cp$, we now have a precise definition of the set of $\Cp$-correlations. However, this definition was phrased in terms of the chosen orderings on each $\inc{v}$ and $\out{v}$; since these are not part of the data of the original graph $G$, we still need to show that the set of $\Cp$-correlations is independent of that choice.

\begin{prop}
The set of $\Cp$-correlations on $G$ does not depend on the particular orderings chosen on the sets $\inc{v}$ and $\out{v}$.
\end{prop}

\begin{sketchproof}
It is enough to show that the ordering on $\inc{v}$ and $\out{v}$ for some fixed $v$ can be modified arbitrarily without changing the set of $\Cp$-correlations. Since any two orderings are related by a permutation, and any permutation can be generated from adjacent transpositions which do nothing but exchange a pair of neighboring edges, it is enough to consider the case that only two neighboring edges $e_1,e_2$ are exchanged at some fixed node $v$, either among the incoming edges or among the outgoing edges. We then need to show that if $P(o[V])$ is a $\Cp$-correlation with respect to the original ordering, then it is also a $\Cp$-correlation with respect to the new ordering. Since $P(o[V])$ is a $\Cp$-correlation with respect to the original ordering by assumption, we know that there exists a labelling of edges by objects $X_e$ of $\Cp$ and of nodes by instruments $f_v$ in $\Cp$ which reproduces the probabilities. Now we can construct a labelling of the graph with the new ordering by replacing the instrument $f_v$ with the instrument given by the composition of each component operation $f_v(o[v])$ with the symmetry (swap) $X_{e_1}\otimes X_{e_2}\longrightarrow X_{e_2}\otimes X_{e_1}$, while retaining all other labels. If $e_1$ and $e_2$ are incoming edges of $v$, then the symmetry has to be applied before $f_v$; otherwise, if $e_1,e_2\in\out{v}$, it should be used after. This defines a new set of valuations of the same graph with the new ordering which reproduces the same probabilities $P(o[V])$ thanks to the invariance properties of the graphical calculus as given in terms of the map $\nu$~\cite[Sec.~2.1]{JS}.
\end{sketchproof}

We would like to reemphasize that having to deal with these orderings explicitly is not a shortcoming of the use of category theory per se; it is rather due to the fact that the appropriate pieces of category theory do not seem to have been developed \emphalt{yet}. Recent work like~\cite{SSR} paves the way towards improving this situation.

If one follows the ``hidden variables everywhere'' paradigm advocated in this paper, then there should be no difference between the causal structure corresponding to direct causal links $v\to w$ and the one corresponding to indirect causal links $v\rsa w$. The intuitive reason is this: clearly, every correlation explainable in terms of hidden variables on the given causal structure $G$ is likewise explainable by the enlarged causal structure comprising all direct \emphalt{and indirect} causal links. To see that the converse also holds, one can simulate the propagation of a hidden variable along an indirect causal link by passing it \emphalt{incognito} along all the direct causal links making up the indirect link as additional information. So there is ultimately no reason to distinguish between a direct link $v\to w$ and an indirect link $v\rsa w$, and so it should be sufficient to define the causal structure only in terms of $v\rsa w$.

\begin{thm}
\label{onlyposet}
The set of $\Cp$-correlations depends only on the partial order ($=$causal set) induced by $G$, as defined in Section~\ref{causalsec}.
\end{thm}

\begin{sketchproof}
It is sufficient to show the following: if $u,v,w\in V$ are nodes such that $u\to v$ and $v\to w$ but $u\not\to w$, then the graph $G'\defin (V,E\cup\{u\to w\})$ differs from $G$ only in that the edge $u\to w$ has been added. It is enough to show that $G'$ has the same set of $\Cp$-correlations as $G$, since applying this statement repeatedly proves the claim.

As noted above, every $\Cp$-correlation on $G$ is clearly also a $\Cp$-correlation on $G'$: any labelling of $G$ by objects and instruments in $\Cp$ can be extended to a labelling of $G'$ by assigning the monoidal unit $I\in\Cp$ to the new edge $u\to w$, while retaining all other labels. Since labelling an edge by $I$ is the same as not considering that edge at all, the resulting correlation is the same.

Conversely, suppose that $P(o[V])$ is a $\Cp$-correlation on $G'$ constructed in terms of some labelling by objects $(X'_e)_{e\in E\cup\{u\to w\}}$ and instruments $(f'_x)_{x\in V}$. We can turn this into a corresponding labelling of $G$ by again retaining all labels, except for the replacements
\[
X_{u\to v} \defin X'_{u\to v} \otimes X'_{u\to w}, \qquad X_{v\to w} \defin X'_{v\to w} \otimes X'_{u\to w}, \qquad f_v(o[v]) \defin f'_v(o[v]) \otimes \id_{X'_{u\to w}},
\]
where we have assumed that the additional edge $u\to w$ comes last in the ordering of $\inc{v}$ and $\out{v}$ on $G'$. Intuitively, the idea is this: in $G'$, the information on $X_{u\to w}$ gets passed directly from $u$ to $w$, while $G$ needs to use the intermediate node $v$ as a relay for sending the information from $u$ to $w$ via $v$. By considering the value map $\nu$ only on the subgraph with boundary induced by the three nodes $u$, $v$ and $w$, the properties~\ref{velem}--\ref{vcomp} imply that the resulting number $P(o[V])$ ends up being the same.
\end{sketchproof}

We now illustrate this categorical formalism with the example of classical correlations, before applying it to quantum correlations in the next section. Classical correlations are defined in terms of a category $\sop$ whose objects are measurable spaces $(X,\Sigma_X)$, and whose morphisms $(X,\Sigma_X)\to (Y,\Sigma_Y)$ are stochastic operations; see Appendix~\ref{appsop}. The monoidal product $\otimes$ of measurable spaces is simply given by their cartesian product. The unit object $I$ is any fixed measurable space having exactly one point. The subcategory of normalized operations $\sp\subseteq\sop$ contains precisely those stochastic operations $P(y|x)$ which are normalized, i.e.~$\int_y P(y|x)=1$ for all $x\in X$. It is straightforward to check that this turns $\sop$ into a symmetric monoidal category satisfying our assumptions~\ref{Ccone}--\ref{CPterm}.

Instantiating the above general definition for the category $\sop$ results in the notion of $\sop$-correlation, which coincides exactly with the notion of classical correlation from Definition~\ref{defccorr}. The integrals in~\eqref{ccorr} arise from the composition of stochastic operations as defined in~\eqref{stochcomp}.

Besides the definition of quantum correlations in the next section, we also suspect that it should be possible to define general probabilistic~\cite{Barrett} correlations as arising from a suitable category in the analogous way. Currently, the problem with defining general probabilistic theories is that a corresponding category of general probabilistic operations does not yet seem to have been worked out. However, it has been shown in independent work of Henson, Lal and Pusey~\cite{HLP} that there exist correlations in certain scenarios that are not general probabilistic correlations. This should be seen in contrast to the case of Bell scenarios, where the correlations correspond to the no-signaling boxes (Theorem~\ref{nsbox}), and it is well-known that every no-signaling box arises from the general probabilistic theory known as ``box world''~\cite{SB}. 

Sometimes a certain framework for physical theories is clearly more powerful than another one. For example, it is commonly believed that everything that can be done classically can also be done quantumly, and this should imply that every classical correlation is also a quantum correlation. One can try to prove statements of this form by making use of this result:

\begin{prop}
\label{DC}
Let $\Cp$ and $\Dp$ be symmetric monoidal category of operations satisfying~\ref{Ccone}--\ref{CPterm}. If there is a linear symmetric monoidal functor $F:\Cp\to\Dp$ with $F(\Pp)\subseteq\Qp$, then every $\Cp$-correlation is also a $\Dp$-correlation.
\end{prop}

Here, linearity of $F$ means that $F$ preserves addition and scalar multiplication of morphisms. The functor is monoidal if it comes equipped with isomorphisms $F(X\otimes_{\Cp} Y) \cong F(X)\otimes_{\Dp} F(Y)$ satisfying the appropriate coherence laws, one of which involves the symmetry isomorphisms. These coherence requirements are technical conditions guaranteed to hold in any naturally arising situation.

\bigskip
\begin{sketchproof}
Let $P(o[V])$ be a $\Cp$-correlation arising via objects $(X_e)_{e\in E}$ and instruments $(f_v)_{v\in V}$. Linearity of the functor $F$ and $F(\Pp)\subseteq\Qp$ imply that the instrument components $f_v(o[v])$ map to $\Dp$-operations $F(f_v(o[v]))$ which form an instrument themselves that we denote $F(f_v)$. Hence we obtain a labelling of $G$ in $\Dp$ with objects $\left(F(X_e)\right)_{e\in E}$ and instruments $\left(F(f_v)\right)_{v\in V}$. The properties~\ref{velem}--\ref{vcomp} then guarantee that the resulting probabilities $P(o[V])$ are the same, since they imply that applying $\nu$ commutes with $F$, and $F$ maps the numbers in $\Cp(I,I)$ to the same numbers in $\Dp(I,I)$ thanks to linearity.
\end{sketchproof}

\newpage
\section{Quantum correlations}
\label{secqc}

Our world is governed, at least on microscopic scales, by the laws of quantum mechanics and quantum field theory. It was the revolutionary insight of Bell~\cite{Bell} that this leads to correlations (on the graph of Figure~\ref{Bellscen}) which are not classical. As a consequence, either our world cannot be described in terms of classical probability theory, or the given causal structure is not actually the correct one describing the given correlations. The first option is generally known as \emph{non-realism}~\cite{Gisin,Gisin2}, while the second option is referred to as \emph{non-locality}, in the sense of the existence of non-local interactions\footnote{Unfortunately, the term \emphalt{nonlocality} is also often used in a much broader sense, as a synonym for the existence of non-classical correlations~\cite{Bellreview}.}. On the mathematical level, the non-realism solution is preferable since non-locality requires fine-tuning to explain why these non-local interactions cannot be exploited for instantaneous signalling~\cite{WS}.

Using the definitions and results from the previous section, we will now define what a \emph{quantum correlation} is in our more general formalism. The basic idea is that instead of carrying measurable spaces as in the classical case, the edges should carry the quantum analogue of those, i.e.~Hilbert spaces. Similarly, instead of having stochastic operations at the nodes, the nodes should now be labelled by \emph{quantum operations}. Roughly speaking, quantum correlations are exactly those joint distributions $P(o[V])$ that one can obtain upon interpreting the given graph in Hardy's operator-tensor formulation of quantum theory~\cite{HardyOT}. However, since the operator-tensor formulation has so far only been worked out for finite-dimensional quantum systems, we cannot use it for our definition of quantum correlation. We prefer to be general and allow systems of arbitrary finite or infinite dimension, since we have no way to tell how much stronger quantum correlations based on infinite-dimensional Hilbert space may be. While we suspect that there should be a quantum analogue of Theorem~\ref{finitethm}, already stating and proving such a theorem requires having a general definition. Without a general definition, there is no way to know whether quantum correlations with infinite-dimensional Hilbert spaces can be stronger than those based on finite dimensions or not.

So in order to apply the definitions and results from the previous section, we need to explain how quantum operations form a symmetric monoidal category, which we denote $\qop$. This works as follows: the objects of $\qop$ are Hilbert spaces $\H$ of arbitrary dimension, finite or infinite. These can be tensored in the usual way; we refer to~\cite{HalmHilb,KR} for background on the definition of this tensor product in the infinite-dimensional case. To every Hilbert space $\H$ we associate its space of trace class operators $\S(\H)$. One can understand the letter ``$\S$'' as suggestive of either ``state'' or of ``Schatten class''~\cite{Conway}. The quantum states (density matrices) on $\H$ are precisely those positive elements of $\S(\H)$ that have unit trace. With this in mind, we can now define the morphisms of $\qop$: for Hilbert spaces $\H$ and $\K$, the morphisms forming the set $\qop(\H,\K)$ are defined to be the \emph{completely positive} maps $\S(\H)\to\S(\K)$. Note that our usage of the term ``quantum operation'' differs from the usual one~\cite[Sec.~8.2.4]{NC} in that we do not impose subnormalization. The normalized maps $\S(\H)\to\S(\K)$ forming the subcategory $\qop_1$ are those quantum operations that are trace-preserving, i.e.~the quantum channels. It is then straightforward to check that our requirements~\ref{Ccone}--\ref{CPterm} are satisfied.

\begin{defn}
$P(o[V])$ is a \emph{quantum correlation} if it is a $\qop$-correlation.
\end{defn}

We refer to the proof of Proposition~\ref{bellclassification} for an example of how to unfold this definition for a concrete graph. Thanks to the general Theorem~\ref{Ccorriscorr} from the previous section, we obtain immediately:

\begin{prop}\label{qisc}
If $P(o[V])$ is a quantum correlation, then it is a correlation.
\end{prop}

Moreover, we also would like to apply Proposition~\ref{DC} in order to show that every classical correlation is also a quantum correlation. However, the problem with this is that we do not know how to construct any functor $\sop\to\qop$ satisfying our requirements, since there is no obvious way to assign a Hilbert space to any given measurable space $(X,\Sigma)$. The reader not interested in the reasons for this may skip the following remark.

\begin{rem}
There is a standard construction assigning a Hilbert space $L^2(X,\Sigma,P)$ to any \emphalt{measure space} $(X,\Sigma,P)$, but this Hilbert space depends crucially on the measure $P$, so that it is not applicable to a mere measurable space $(X,\Sigma)$. In fact, even for the full subcategory $\fsop\subseteq\sop$ containing only the finite measurable spaces $(X,\mathbf{2}^X)$, we have not been able to find any functor $F : \fsop\to\qop$ that would satisfy the assumptions of Proposition~\ref{DC}. Any obvious candidate should be given on objects by assigning to every finite set $X$ the Hilbert space $\C^X$, which has a canonical orthonormal basis given by $\{|x\rangle\}_{x\in X}$. To any stochastic operation $P(y|x) : X\to Y$, we can try to assign the quantum operation
\beq
\label{fun1}
\S(\C^X) \to \S(\C^Y), \qquad |x_1\rangle\langle x_2| \mapsto \sum_{y_1,y_2} \sqrt{P(y_1|x_1)\, P(y_2|x_2)} \: |y_1\rangle\langle y_2| ,
\eeq
which can be understood as being represented by the single Kraus operator
\[
 \C^X \to \C^Y, \qquad |x\rangle \mapsto \sum_y \sqrt{P(y|x)} \: |y\rangle .
\]
It is straightforward to check that this assignment of a quantum operation to any stochastic operation is indeed a monoidal functor. However, due to the square root in~\eqref{fun1}, it fails to be linear---the sum of two classical operations will in general not map to the sum of the associated quantum operations! So, alternatively, one might try and come up with an alternative to~\eqref{fun1} which is linear in $P$. The obvious candidate is a ``decoherence in the canonical basis'' approach, in which we would assign to $P(y|x)$ the quantum operation
\beq
\label{fun2}
\S(\C^X) \to \S(\C^Y), \qquad |x_1\rangle\langle x_2| \mapsto \delta_{x_1,x_2} \sum_{y} P(y|x_1) \: |y\rangle\langle y| ,
\eeq
which indeed has the advantage of being linear in $P$. However, this assignment lacks functoriality: the identity stochastic operation $X\to X$, represented by $P(x'|x)=\delta_{x',x}$, gets mapped to the quantum operation $\S(\C^X)\to\S(\C^X)$ which decoheres any quantum state with respect to the computational basis; in other words, this assignment of a quantum operation to any stochastic operation does not map identities to identities, as any functor should, and hence~\eqref{fun2} does not define a functor.
\end{rem}

So while we strongly suspect that every classical correlation is indeed a quantum correlation, we do not have any proof of this seemingly basic statement! 

\begin{conj}
\label{cisq}
Every classical correlation is also a quantum correlation.
\end{conj}

All that we can say for sure is that every classical correlation which can be realized with finite hidden variable spaces is also a quantum correlation. Providing a proof sketch of this is what we do now. By Theorem~\ref{finitethm}, this is arbitrarily close to proving the full conjecture, but does not go all the way. 

\medskip
\begin{partialproof}
As outlined above, we do not know how to apply Proposition~\ref{DC} to this situation, but turning any stochastic operation between finite sets into a quantum operation as in~\eqref{fun2} does indeed do the trick, despite not being functorial. So for each edge $e\in E$, we replace the hidden variable space $X_e$, assumed to be finite, by the Hilbert space $\C^{X_e}$. We also replace every stochastic instrument $P(o[v]\lambda_{\out{v}}|\lambda_{\inc{v}})$ by the corresponding ``decoherence in the canonical basis'' quantum instrument
\[
O_v\times \S\Bigg(\bigotimes_{e\in\inc{v}} \C^{X_e} \Bigg) \longrightarrow \S\Bigg(\bigotimes_{e\in\out{v}} \C^{X_e} \Bigg)
\]
with same outcome set $O_v$. Then all states propagating along the edges will be diagonal in the canonical basis. This guarantees that the original joint probabilities $P(o[V])$ are recovered.
\end{partialproof}

In order to underline that our definition of quantum correlation is adequate, we compare it to the usual definition in the Bell scenario case:

\begin{prop}\label{bellclassification}
In the $n$-party Bell scenario, $P(o[a_1\ldots a_n x_1\ldots x_n s])$ is a quantum correlation if and only if it satisfies the free will equation~\eqref{freewill} and the corresponding conditional probabilities admit a \emph{quantum model},
\beq
\label{qmodel}
P(o[a_1\ldots a_n] | o[x_1 \ldots x_n s]) = \braket{\psi_{o[s]}|E^{o[x_1]}_{o[a_1]} \otimes \dotsb \otimes E^{o[x_n]}_{o[a_n]}|\psi_{o[s]}},
\eeq
where $\{|\psi_{o[s]}\rangle\}_{o[s]\in O_s}$ is a family of states on some composite Hilbert space $\H_1\otimes\ldots\otimes\H_n$ and the $E_{o[a_i]}^{o[x_i]}$ are positive operators on $\H_i$, not depending on $o[s]$, subject to the completeness relation $\sum_{o[a_i]} E_{o[a_i]}^{o[x_i]} = \mathbbm{1}_{\H_i}$.
\end{prop}

As in Proposition~\ref{cbell}, we implicitly also assume that only those values of $o[x_i]$ and $o[s]$ are considered for which $P(o[x_i]) > 0$ and $P(o[s]) > 0$, since otherwise conditioning on these variables would not make sense.

\begin{proof}
We begin by spelling out our definition of quantum correlation for the $n$-party Bell scenario in more detail.
First, our definition requires that the edges of the graph are labelled by Hilbert spaces
\[
\H_{x_i\to a_i},\qquad \H_{s\to a_i} .
\]
Second, there is a quantum instrument labelling the source node $s$. Since $s$ does not have any incoming edges, this quantum instrument corresponds to a collection of subnormalized states indexed by the outcome $o[s]\in O_s$,
\[
\rho_{o[s]} \:\in\: \S\Bigg(\bigotimes_{i=1}^n \H_{s\to a_i} \Bigg), \qquad \tr\left(\sum_{o[s]} \rho_{o[s]}\right) = 1,
\]
where the prescription in the proof of Theorem~\ref{Ccorriscorr} shows that $P(o[s]) = \tr(\rho_{o[s]})$. The same statement applies to the ``choice of setting'' nodes $x_i$, which likewise emit subnormalized states indexed by outcomes $o[x_i]$,
\[
\sigma_{o[x_i]} \:\in\: \S(\H_{x_i \to a_i}), \qquad \tr\left(\sum_{o[x_i]} \sigma_{o[x_i]}\right) = 1,
\]
and we similarly have $P(o[x_i]) = \tr(\sigma_{o[x_i]})$.

Finally, the quantum instrument labelling a ``measurement'' node $a_i$ corresponds simply to a POVM\footnote{To begin with, the components of the quantum instrument at $a_i$ are completely positive maps of type $\S(\H)\to\C$ for $\H=\H_{x_i\to a_i} \otimes \H_{s\to a_i}$, but it is well-known that these are in bijection with bounded positive operators on $\H$.},
\[
M_{o[a_i]} \:\in\: \B\left( \H_{x_i\to a_i} \otimes \H_{s\to a_i}\right) ,\qquad \sum_{o[a_i]} M_{o[a_i]} = \mathbbm{1}.
\]
This data is related to the overall probabilities by
\beq
\label{newqmodel}
P(o[a_1\ldots a_n x_1\ldots x_n s]) = \tr\left[\left(\rho_{o[s]} \otimes \bigotimes_{i=1}^n \sigma_{o[x_i]} \right) \left(\bigotimes_{i=1}^n M_{o[a_i]}\right)\right],
\eeq
where it is understood that the tensor factors in the tensor product Hilbert space carrying the ``state'' part need to be reordered so as to match those of the ``measurement'' part in order for the operator multiplication to make sense. (This is closely related to the ordering issues that we encountered in Section~\ref{seccats}.)

We begin the proof with the ``only if'' part. One obtains the free will equation~\eqref{freewill} by noting that Proposition~\ref{qisc} guarantees that $P(o[a_1\ldots a_nx_1\ldots x_ns])$ is a correlation, and therefore Theorem~\ref{nsbox} applies. In order to obtain a representation of the form~\eqref{qmodel}, we take the Hilbert spaces to be $\H\defin\H_{s\to a_i}$, on which we have positive operators
\beq
\label{defE}
E_{o[a_i]}^{o[x_i]} \defin \frac{1}{P(o[x_i])} \,\tr_{\H_{x_i\to a_i}} \left( \left(\sigma_{o[x_i]}\otimes\mathbbm{1}_{\H_{s\to a_i}}\right) M_{o[a_i]} \right).
\eeq
Since $\sum_{o[a_i]} M_{o[a_i]} = \mathbbm{1}$ and $\tr(\sigma_{o[x_i]}) = P(o[x_i])$, this indeed satisfies the completeness relation $\sum_{o[a_i]} E_{o[a_i]}^{o[x_i]} = \mathbbm{1}_{\H_{x_i\to a_i}}$. The definition~\eqref{defE} formalizes the intuitive idea that $o[x_i]$ behaves like a choice of setting determining which POVM will be used at $a_i$. With this, we can rewrite~\eqref{newqmodel} in the form
\[
P(o[a_1\ldots a_n x_1 \ldots x_n s]) = \tr\left[\left(\rho_{o[s]} \otimes \bigotimes_{i=1}^n \mathbbm{1}_{\H_{s\to a_i}} \right)\left(\bigotimes_{i=1}^n E_{o[a_i]}^{o[x_i]}\right)\right] \cdot\prod_{i=1}^n P(o[x_i]),
\]
where now the trace only ranges over the remaining Hilbert space $\bigotimes_{i=1}^n \H_{s\to a_i}$. If we now purify the states $\rho_{o[s]}$ jointly to unit vectors $|\psi_{o[s]}\rangle$ by adding an additional ancilla system, which should be regarded as becoming part of any one of the $\H_{s\to a_i}$, we can write this as
\beq
\label{altqmodel}
P(o[a_1\ldots a_n x_1 \ldots x_n s]) = \left\langle\psi_{o[s]}\left|E_{o[a_1]}^{o[x_1]}\otimes\ldots\otimes E_{o[a_n]}^{o[x_n]} \right|\psi_s\right\rangle \cdot P(o[x_1])\cdots P(o[x_n]) P(o[s]),
\eeq
where we had to factor out an additional factor of $P(o[s])$, since $|\psi_{o[s]}\rangle$ is normalized, while $\tr(\rho_{o[s]}) = P(o[s])$. This gives the desired form~\eqref{qmodel} upon dividing by the free will equation~\eqref{freewill}.

Now for the ``if'' part. Suppose~\eqref{freewill} and~\eqref{qmodel} hold; we can multiply these two equations and obtain~\eqref{altqmodel}. We rewrite this in the form~\eqref{newqmodel} as follows. Let the Hilbert spaces on the edges be given by
\[
\H_{s\to a_i} \defin \H_i ,\qquad \H_{x_i\to a_i} \defin \C^{O_{x_i}},
\]
where $\C^{O_{x_i}}$ is the Hilbert space with a canonical orthonormal basis $\{|o[x_i]\rangle\}$ indexed by the possible values of $o[x_i]$. Now define
\[
\rho_{o[s]} \defin P(o[s]) \cdot |\psi_{o[s]}\rangle\langle\psi_{o[s]}|,\quad \sigma_{o[x_i]} \defin P(o[x_i]) \cdot |o[x_i]\rangle\langle o[x_i]|, \quad M_{o[a_i]} \defin \sum_{o[x_i]} |o[x_i]\rangle\langle o[x_i]| \otimes E_{o[a_i]}^{o[x_i]}.
\]
It is straightforward to check that these operators satisfy the properties mentioned above that are required for the operators to assemble into quantum instruments on the nodes $s$, $x_i$ and $a_i$, and that~\eqref{newqmodel} then indeed recovers the given correlation.
\end{proof}

The proofs of this proposition and the earlier Propositions~\ref{nsbox} and~\ref{cbell} should make it clear how to translate back and forth between Bell scenarios in our formalism and the already existing vast literature on Bell scenarios, which uses the ordinary definitions~\eqref{lcausal} and~\eqref{qmodel}.

We leave all further study of quantum correlations on arbitrary causal structure to future work. As a start, it should be verified that our general definition not only comprises the usual one for Bell scenarios, but also for all other concrete scenarios that have been studied until now~\cite{Pop,BGP,Fri,GWCAN}, as far as precise definitions have been spelled out in these works. Then, it will be a challenging enterprise to find further examples of non-classical quantum correlations and try to answer the question in which scenarios non-classical quantum correlations exist and in which ones they do not; as a trivial example, the single-party Bell scenario does not have any non-classical quantum correlations. We believe that a lot remains to be explored, and works like the present one and~\cite{Fri,HLP} are but the beginning of a transformation of our understanding of quantum ``nonlocality'' and its applications to quantum information processing.

\newpage
\section{Classical correlations as hidden Bayesian networks}
\label{hbn}

\subsection*{Hidden Markov models.} Hidden variables have been studied extensively not only in quantum foundations, but also in machine learning, and in particular hidden variables for stochastic processes. Recall that a \emph{stochastic process} is a joint distribution $P(o[\Z])$ of a collection of random variables $o[\Z]=(o[t])_{t\in\Z}$ indexed by integer timesteps $t\in\Z$. The underlying causal structure is the one of $\Z$, as in Figure~\ref{stochproc}; the node labels are precisely the timesteps $t\in\Z$. We now informally discuss correlations and classical correlations on this causal structure; this must remain at an informal level since our definitions and results of the previous sections, like Definition~\ref{defccorr}, are not actually applicable to the case of infinite $G$, like $G=\Z$. 

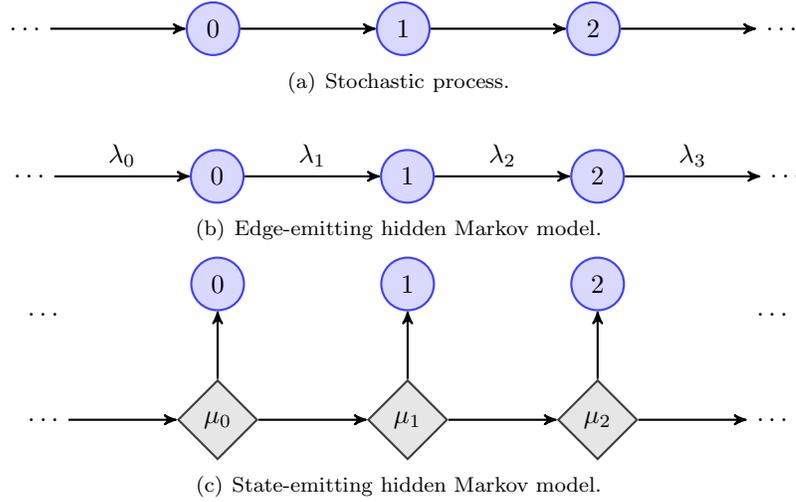
\begin{figure}
\subfigure[Stochastic process.]{
\label{stochproc}
\begin{tikzpicture}[node distance=1.8cm,>=stealth',thick,scale=.8]
\tikzstyle{place}=[circle,thick,draw=blue!75,fill=blue!15,minimum size=7mm]
\node[place] (a-1) at (0,0) {$0$} ;
\node[place] (a0) [right=of a-1] {$1$} ;
\node[place] (a1) [right=of a0] {$2$} ;
\node (a-2) [left=of a-1] {$\ldots$} ;
\node (a2) [right=of a1] {$\ldots$} ;
\draw[->] (a-2) -- (a-1) ;
\draw[->] (a-1) -- (a0) ;
\draw[->] (a0) -- (a1) ;
\draw[->] (a1) -- (a2) ;
\end{tikzpicture}}
\subfigure[Edge-emitting hidden Markov model.]{
\rule{0ex}{3pc}
\label{eehmm}
\begin{tikzpicture}[node distance=1.8cm,>=stealth',thick,scale=.8]
\tikzstyle{place}=[circle,thick,draw=blue!75,fill=blue!15,minimum size=7mm]
\node[place] (a-1) at (0,0) {$0$} ;
\node[place] (a0) [right=of a-1] {$1$} ;
\node[place] (a1) [right=of a0] {$2$} ;
\node (a-2) [left=of a-1] {$\ldots$} ;
\node (a2) [right=of a1] {$\ldots$} ;
\draw[->] (a-2) -- node[midway,above] {$\lambda_0$} (a-1) ;
\draw[->] (a-1) -- node[midway,above] {$\lambda_1$} (a0) ;
\draw[->] (a0) -- node[midway,above] {$\lambda_2$} (a1) ;
\draw[->] (a1) -- node[midway,above] {$\lambda_3$} (a2) ;
\end{tikzpicture}}
\subfigure[State-emitting hidden Markov model.]{
\rule{0ex}{6pc}
\label{sehmm}
\begin{tikzpicture}[node distance=1.44cm,>=stealth',thick,scale=.8]
\tikzstyle{place}=[circle,thick,draw=blue!75,fill=blue!15,minimum size=7mm]
\tikzstyle{transition}=[diamond,thick,draw=black!75,fill=black!10,minimum size=6.0mm]
\node[transition] (l0) at (0,0) {$\mu_1$} ;
\node[transition] (l-1) [left=of l0] {$\mu_0$} ;
\node[transition] (l1) [right=of l0] {$\mu_2$} ;
\node[place] (a0) [node distance=.9cm,above=of l0] {$1$} ;
\node[place] (a-1) [node distance=.9cm,above=of l-1] {$0$} ;
\node[place] (a1) [node distance=.9cm,above=of l1] {$2$} ;
\node (l-2) [left=of l-1] {$\ldots$} ;
\node (a-2) [node distance=1.1cm,above=of l-2] {$\ldots$} ;
\node (l2) [right=of l1] {$\ldots$} ;
\node (a2) [node distance=1.1cm,above=of l2] {$\ldots$} ;
\draw[->] (l-2) -- (l-1) ; 
\draw[->] (l-1) -- (l0) ;
\draw[->] (l0) -- (l1) ;
\draw[->] (l1) -- (l2) ;
\draw[->] (l-1) -- (a-1) ;
\draw[->] (l0) -- (a0) ;
\draw[->] (l1) -- (a1) ;
\end{tikzpicture}}
\caption{A stochastic process and the two kinds of hidden Markov models.}
\end{figure}

Informally speaking, every joint distribution $P(a_{\Z})$ should be a correlation in the sense of Definition~\ref{defcor}, since there are no two non-empty subsets of $\Z$ with disjoint causal past. Realizations of $P(a_\Z)$ in terms of classical hidden variables, as in Definition~\ref{defccorr} and depicted in Figure~\ref{eehmm}, are known as \emph{edge-emitting hidden Markov models}~\cite[Sec.~3.2]{Upper}\footnote{There are other technical differences besides the fact that our definition only applies to finite graphs. The hidden variable spaces $X_t$ in a hidden Markov model are typically assumed to be finite, and the stochastic maps $P(a_t\lambda_{t+t}|\lambda_t)$ are typically assumed to be independent of $t$. We use the term ``hidden Markov model'' in a sense which does not impose either of these requirements.}. Hidden Markov models have a wide range of applications e.g.~in machine learning~\cite{GY} and computational biology~\cite{Eddy}. The essential idea is that the hidden variables $(\lambda_t)_{t\in\Z}$ form a Markov chain in the sense that $\lambda_{t+1}$ may become correlated with $\lambda_{t-1}$ only through $\lambda_t$, while the analogous statement may not be true for the observed values $o[t]$. Since Markov chains such as the hidden distribution $P(\lambda_\Z)$ are much easier to handle mathematically and computationally, it is useful to consider a representation of $P(o[\Z])$, the actual stochastic process of interest, as a classical correlation. It is easy to see that every $P(o[\Z])$ actually admits such a representation by defining each hidden variable space $(X_t,\Sigma_t)$ to be equal to $\prod_{t\in\Z} O_t$, so that one can take all hidden variables $\lambda_t$ to be equal and to coincide with the entire sequence of outputs $(o[t])_{t\in\Z}$; an information processing gate $P(o[t]\lambda_{t+1}|\lambda_t)$ then simply takes an incoming $\lambda_t = o[\Z]$ and deterministically outputs the corresponding component $o[t]$ and forwards the entire sequence $o[\Z]$ as $\lambda_{t+1}$. On the other hand, in general it is \emphalt{not} the case that every $P(o[\Z])$ can be represented as a classical correlation in such a way that all hidden variable spaces $X_t$ are finite~\cite[Sec.~3.6]{Upper}.

So much for edge-emitting hidden Markov models. One more commonly encounters \emph{state-emitting} hidden Markov models~\cite{Ghahr}, of which Figure~\ref{sehmm} shows an illustration. These are defined in a different way: again, there are hidden variable spaces $(Y_t,\Omega_t)$ carrying hidden variables $\mu_t\in Y_t$, but this time they come equipped with a pair of conditional distributions
\[
P(\mu_{t+1}|\mu_t), \qquad P(o[t]|\mu_t),
\]
which can now be thought of as living on the \emphalt{nodes} in Figure~\ref{sehmm}. The first of these describes an updating of the hidden variable from $\mu_t$ to $\mu_{t+1}$, while the second can be thought of as a probing of the hidden variable $\mu_t$ resulting in an outcome $o[t]$. The main point here is that these two updating operations are taken to be independent.

The purpose of this section is to generalize the latter notion of state-emitting hidden Markov model to an arbitrary (finite directed acyclic) graph $G=(V,E)$, which results in the notion of \emph{hidden Bayesian network}. Then we prove that a correlation $P(o[V])$ is classical if and only if it admits a description in terms of a hidden Bayesian network. Figure~\ref{Popscenhbn} displays the structure of a hidden Bayesian network for the example causal structure of Figure~\ref{Popscen}.

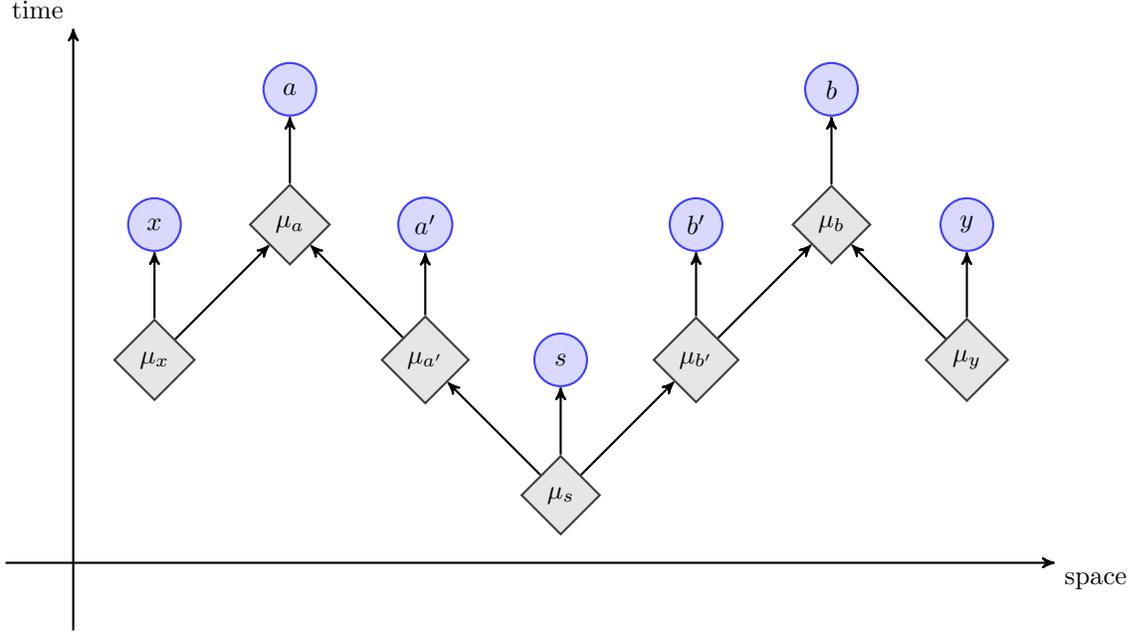
\begin{figure}
\centering
\begin{tikzpicture}[node distance=1.8cm,>=stealth',thick,scale=.9]
\draw[->] (-2.2,-1) -- (-2.2,7.9) node[anchor=south east] {time} ;
\draw[->] (-3.2,0) -- (12.3,0) node[anchor=north west] {space} ;
\tikzstyle{place}=[circle,thick,draw=blue!75,fill=blue!15,minimum size=7mm]
\tikzstyle{transition}=[diamond,thick,draw=black!75,fill=black!10,minimum size=6.0mm]
\node[transition] (a1) at (3,3) {$\mu_{a'}$} ;
\node[transition] (b1) at (7,3) {$\mu_{b'}$} ;
\node[transition] (l) at (5,1) {$\mu_s$} ;
\node[transition] (la) at (-1,3) {$\mu_x$} ;
\node[transition] (lb) at (11,3) {$\mu_y$} ;
\node[transition] (a2) at (1,5) {$\mu_a$} ;
\node[transition] (b2) at (9,5) {$\mu_b$} ;
\node[place, above of = a1] (aa1) at (3,3) {$a'$} ;
\node[place, above of = ab1] (ab1) at (7,3) {$b'$} ;
\node[place, above of = al] (al) at (5,1) {$s$} ;
\node[place, above of = ala] (ala) at (-1,3) {$x$} ;
\node[place, above of = alb] (alb) at (11,3) {$y$} ;
\node[place, above of = aa2] (aa2) at (1,5) {$a$} ;
\node[place, above of = ab2] (ab2) at (9,5) {$b$} ;
\draw[->] (l) -- (a1) ;
\draw[->] (l) -- (b1) ;
\draw[->] (la) -- (a2) ;
\draw[->] (lb) -- (b2) ;
\draw[->] (a1) -- (a2) ;
\draw[->] (b1) -- (b2) ;
\draw[<-] (aa1) -- (a1) ;
\draw[<-] (aa2) -- (a2) ;
\draw[<-] (ab1) -- (b1) ;
\draw[<-] (ab2) -- (b2) ;
\draw[<-] (al) -- (l) ;
\draw[<-] (ala) -- (la) ;
\draw[<-] (alb) -- (lb) ;
\end{tikzpicture}
\caption{Hidden Bayesian network for Popescu's ``hidden nonlocality'' scenario (Figure~\ref{Popscen}).}
\label{Popscenhbn}
\end{figure}

\begin{defn}
A \emph{Bayesian network} on a graph $G$ is an assignment of a measurable space $(Y_v,\Omega_v)$ to every node $v\in V$ together with a conditional probability distribution
\[
P(\mu_v | \mu_{\pa{v}})
\]
which generates the overall distribution
\beq
\label{bn}
P(\mu_V) = \prod_{v\in V} P(\mu_v|\mu_{\pa{v}}).
\eeq
\end{defn}

Bayesian networks have been studied extensively for the case that the measurable spaces $(Y_v,\Omega_v)$ are finite~\cite{Pearl}, and for this case there are several equivalent conditions known for when a joint distribution $P(\mu_V)$ constitutes a Bayesian network.

\begin{defn}
\label{defhbn}
A correlation $P(o[V])$ \emph{arises from a hidden Bayesian network} if there is a Bayesian network with distribution $P(\mu_V)$ on $G$ and stochastic maps $P(o[v]|\mu_v)$ for every $v\in V$ such that
\beq
\label{hbnmodel}
P(o[V]) = \int_{\mu_V} \left(\prod_{v\in V} P(o[v]|\mu_v) \right) \cdot P(\mu_V) .
\eeq
\end{defn}

Intuitively, this formula expresses the idea that each $o[v]$ should be a noisy coarse-graining of the hidden variable $\mu_v$, which are hidden variables whose dependencies are captured by the structure of the graph $G$.

\begin{thm}
\label{hbnthm}
A correlation $P(o[V])$ arises from a hidden Bayesian network if and only if it is classical.
\end{thm}

The proof will actually show that any joint distribution represented in the form~\eqref{hbnmodel} is a classical correlation; in other words, assuming that $P(o[V])$ is a correlation in Definition~\ref{defhbn} is redundant.

\begin{proof}
We start with the ``if'' part, i.e.~we assume a representation as a classical correlation in the form~\eqref{ccorr}. We define the new hidden variable spaces $(Y_v,\Omega_v)$ to be given by, in terms of the old hidden variable spaces $(X_e,\Sigma_e)$,
\beq
\label{Yv}
Y_v \defin O_v \times \prod_{e\in\ch{v}} X_e,
\eeq
where it is understood that the $\sigma$-algebra $\Omega_v$ is given by the product of the corresponding $\Sigma_e$ with $\mathbf{2}^{O_v}$. The basic idea here is that the new hidden variable $\mu_v\in Y_v$ will contain all the data which, in the representation as a classical correlation, corresponds to the output of the information processing gate $P(o[v]\lambda_{\out{v}}|\lambda_{\inc{v}})$. The stochastic map $P(o[v]|\mu_v)$ which assigns to each new hidden variable a corresponding outcome can then simply be taken to be the projection from~\eqref{Yv} onto $O_v$, and as such it is deterministic.

This identification $\mu_v = (o[v],\lambda_{\out{v}})$ of the new hidden variable $\mu_v$ with the outcome $o[v]$ together with the tuple of old hidden variables $\lambda_{\out{v}}$, which together form the output of the old information processing gate $P(o[v]\lambda_{\out{v}}|\lambda_{\inc{v}})$, suggests that we should also choose the conditional distributions in~\eqref{bn} to coincide with that information processing gate,
\[
P(\mu_v | \mu_{\pa{v}}) \defin P(o[v] \lambda_{\out{v}} | \lambda_{\inc{v}}).
\]
This is a sensible definition if we consider $\lambda_{\inc{v}}$ as a function of $\mu_{\pa{v}} = (o[\pa{v}],\lambda_{\out{\pa{v}}})$, so that the left-hand side becomes a well-defined probability measure on~\eqref{Yv} once a concrete tuple of values $\mu_{\pa{v}}$ has been specified.

With this definition of the hidden Bayesian network, we need to show that~\eqref{hbnmodel} indeed reproduces the original correlations. We have
\begin{align*}
P_{\mathrm{new}}(o[V]) &= \int_{\mu_V} \left(\prod_{v\in V} P(o[v]|\mu_v) \right) \cdot P(\mu_V) \stackrel{\eqref{bn}}{=} \int_{\mu_V} \prod_{v\in V} P(o[v]|\mu_v) \, P(\mu_v|\mu_{\pa{v}})  \\
 & \stackrel{(*)}{=} \int_{\lambda_E} \prod_{v\in V} P(o[v] \lambda_{\out{v}} | \lambda_{\inc{v}} ) = P_{\mathrm{old}} (o[V]),
\end{align*}
where $(*)$ has used the fact that $P(o[v]|\mu_v)$ is simply given by projection onto the first factor, so that when decomposing $\mu_v$ into its components, only those tuples with the correct $o[v]$ component survive. Since every edge occurs in $\out{v}$ for exactly one $v$, the resulting integral over $\mu_V$ turns into an integral over $\lambda_E$.

For the ``only if'' part, we start with a representation of the form~\eqref{hbnmodel}. In this case, the new hidden variables are the $\lambda_e$, while the old ones are the $\mu_v$. We define the new hidden variable spaces to be given by
\beq
\label{Xe}
X_e \defin Y_{\src{e}}.
\eeq
In other words, all outgoing edges $e\in\out{v}$ at a given node $v$ carry the same hidden variable space. Moreover, we construct the information processing gate at $v$ so that all the outgoing information is the same on all these edges:
\beq
\label{defgate}
P(o[v]\lambda_{\out{v}}|\lambda_{\inc{v}}) \defin P(o[v] | \mu_{\pa{v}}) \, P(\mu_v|\mu_{\pa{v}}),
\eeq
which is to be understood as follows: the known incoming values $\lambda_{\inc{v}}$ can be identified by~\eqref{Xe} with values for the $\mu_{\pa{v}}$. Then when keeping these values fixed, the right-hand side of~\eqref{defgate} defines a probability measure on $O_v\times Y_v$ (actually a product measure). We can push this forward along the map $Y_v\to X_{\out{v}}$, which by~\eqref{Xe} we can take to be the inclusion along the diagonal. In other words, the resulting measure is supported on those tuples $\lambda_{\out{v}}$ for which all components are the same. This ends the definition of the hidden variables and information processing gates.

We now verify that the resulting correlation is indeed the expected one,
\[
P_{\mathrm{new}}(o[V]) =  \int_{\lambda_E} \prod_{v\in V} P(o[v] \lambda_{\out{v}} | \lambda_{\inc{v}}) = \int_{\lambda_E} \prod_{v\in V} P(o[v]|\mu_{\pa{v}}) \, P(\mu_v|\mu_{\pa{v}}),
\]
where we simply plugged in~\eqref{defgate}, and it is understood that the intended interpretation of~\eqref{defgate} needs to be applied. We can evaluate this further by noting that, at every node $v$, the integral over $\lambda_{\out{v}}$ is supported on the diagonal, and hence it is sufficient to do an integral over $\mu_v$. This proves that the previous expression is equal to
\[
\int_{\mu_V} \prod_{v\in V} P(o[v]|\mu_{\pa{v}}) \, P(\mu_v|\mu_{\pa{v}}) = P_{\mathrm{old}}(o[V]),
\]
as was to be shown.
\end{proof}

It should be interesting to figure out whether a similar statement can be derived for the quantum case, by using e.g.~the definition of ``quantum Bayesian network'' of Leifer and Spekkens~\cite{LS}.

\appendix

\newpage
\section{Measure-theoretic technicalities}
\label{appsop}

\subsection*{Measurable spaces and measures.} 

Here, we discuss the proper mathematics needed to describe classical hidden variables. A classical hidden variable is modelled by a set of possible values equipped with a probability distribution. Since we would like to allow this set of values to be of any size, finite or infinite, we first need to recall the standard mathematical structures used for dealing with probability distributions with an arbitrary number of outcomes. This is standard material can be found in any textbook on measure theory, e.g.~\cite{HalmMeas}. The notions that we recall in the following can be thought of as classical analogues of concepts which are well-known to the quantum information and foundations community in the quantum-mechanical context. 

\begin{defn}
\label{defmessp}
A \emph{measurable space} is a pair $(X,\Sigma)$ consisting of a set $X$ and a collection of subsets $\Sigma\subseteq 2^X$ such that:
\begin{itemize}
\item $\emptyset \in \Sigma$,
\item If $U\in\Sigma$, then also $X{\setminus} U\in\Sigma$,
\item If $U_1,U_2,\ldots \in \Sigma$, then also $\bigcup_i U_i \in\Sigma$.
\end{itemize}
\end{defn}

If $U\in\Sigma$, then $U$ is called a \emph{measurable set}. The measurable sets $U$ are those to which one expects to be able to assign a number measuring its ``size'', like the probability that a randomly chosen element of $X$ will lie in $U$. A collection of sets $\Sigma$ satisfying these axioms is called a \emph{$\sigma$-algebra}.

We always assume that $\Sigma$ \emphalt{separates points} in $X$, which means that for any $x,y\in X$ there is a measurable set $A\in\Sigma$ such that $x\in A$, but $y\not\in A$. If $X$ is finite, this is equivalent to assuming that $\Sigma=2^X$, meaning that all sets are actually measurable. The necessity to distinguish measurable from non-measurable sets arises only for infinite $X$.

The r\^ole of measurable spaces in classical probability theory is more or less analogous to the one of Hilbert spaces in quantum theory. The analogues of states are the measures:

\begin{defn}
A \emph{measure} on $(X,\Sigma)$ is a function $P:\Sigma\to\Rplus$ with 
\begin{itemize}
\item $P(\emptyset)=0$,
\item $P\left(\bigcup_i U_i\right) = \sum_i P(U_i)$ for every sequence of pairwise disjoint measurable sets $U_1,U_2,\ldots\in\Sigma$.
\end{itemize}
\end{defn}

The \emph{normalization} of a measure is the number $P(X)$. If $P(X)=1$, then $P$ is a \emph{probability measure}; this models a situation in which it is certain that some outcome $x\in X$ occurs, and then we say that $(X,\Sigma,P)$ is a \emph{probability space}. Since we do not always want to make this assumption, we also allow that the total probability may be different from one, in which case it will typically be less than one. We nevertheless use the letter ``$P$'' to denote a measure since all our measures do in some way or the other have an interpretation as probabilities.

The r\^ole of an operator on a Hilbert space is played by the concept of measurable function:

\begin{defn}
A \emph{measurable function} on $(X,\Sigma)$ is a function $f:X\to \R$ such that for all $a,b\in\R$, the set
\[
f^{-1}([a,b]):=\{ x\in X \:|\: a \leq f(x) \leq b \}
\]
is measurable, i.e.~in $\Sigma$.
\end{defn}

Measurable functions that are \emph{bounded}, i.e.~take values in a bounded subset of $\R$, can be integrated with respect to any measure $P$. This means that we can form the integral
\[
\int_x f(x) P(x),
\]
and this corresponds to the \emph{expectation value} of $f$. Note that we have now written the integration variable formally as an argument of $P$. Throughout this paper, this should be understood to be symbolic notation for integral of a function $f$ with respect to a measure $P$. It is symbolic since the measure $P(\{x\})$ of any one-element set $\{x\}$ is typically zero, and hence it is not meaningful in any sense to regard the measure $P$ as a function of the variable $x$. For the precise definition of this integral we refer to the literature on measure theory~\cite{HalmMeas}.

\subsection*{Stochastic operations.} In the main text, we also crucially need to deal with operations taking a collection of hidden variables as input and producing another collection of hidden variables as output. In order to define what exactly this is supposed to mean, we need to define \emph{stochastic operations} between measurable spaces, which are the classical analogue of quantum operations. To this end, We are going to reproduce a variant of a definition of Panangaden~\cite{PP}; the difference is that we allow measures with arbitrary normalization instead of only subnormalized ones, and we also speak of stochastic operations instead of ``stochastic relations'' in order to keep the terminology more in line with the standard jargon for the quantum case.

\begin{defn}
Let $(X,\Sigma_X)$ and $(Y,\Sigma_Y)$ be measurable spaces. A \emph{stochastic operation} from $X$ to $Y$ is a function
\[
P(y|x) \: : \: X \times \Sigma_Y \to [0,1]
\]
satisfying the following properties:
\begin{itemize}
\item For every fixed $x\in X$, $P(y|x)$ is a measure on $Y$.
\item For every fixed $U\in\Sigma_Y$, the function $P(U|x)$ is a bounded measurable function of $x$.
\end{itemize}
\end{defn}

The idea behind the notation ``$P(y|x)$'' is that for every $x\in X$, one obtains a measure on $Y$ which, by abuse of notation, we also denote $P(y|x)$. If both $X$ and $Y$ are finite, $P(y|x)$ is simply a matrix of nonnegative numbers indexed by $x\in X$ and $y\in Y$. In any case, a stochastic operation, with our definition, has the property that the normalization of the measure $P(y|x)$ may depend on the value of $x$.

There is an important special case of a stochastic operation: if $I$ is a one-element set, then the stochastic operations from $I$ to $Y$ are precisely the measures on $Y$. This is similar to how quantum states can themselves be regarded as maps from a trivial system to a non-trivial system~\cite{CP}.

A stochastic operation from $X$ to $Y$ can be composed with a stochastic operation from $Y$ to $Z$ to give a stochastic operation from $X$ to $Z$:
\beq
\label{stochcomp}
P(z | x) \defin \int_y P(z | y) P(y | x).
\eeq
In words, this means that for every fixed $z$, we integrate the function $y\mapsto P(z | y)$ with respect to the measure $P(y|x)$. We refer to~\cite{PP} for the (non-trivial!) proof showing that this definition of composition of stochastic operations does indeed define a category, and we denote this category by $\sop$.

There is an interesting subcategory $\sp\subseteq\sop$ consisting of those stochastic operations $P(y|x)$ which are \emph{normalized}, by which we mean that they satisfy
\[
\int_y P(y|x) = 1
\]
for all $x\in X$. For such a normalized stochastic operation, we also use the term \emph{conditional distribution}.

If one only considers finite measurable spaces and normalized stochastic operations between those, one obtains a category which coincides with the category $\FinStoch$ discussed e.g.~in~\cite{BF}.

\subsection*{Products of measure spaces and stochastic operations.}

For given measurable spaces $(X,\Sigma_X)$ and $(Y,\Sigma_Y)$, there is a way to form their \emph{product space} $(X\times Y,\Sigma_X\otimes\Sigma_Y)$, where $X\times Y$ is the cartesian product. The $\sigma$-algebra $\Sigma_X\otimes\Sigma_Y$ is defined to contain all sets of the form $U\times V$ for $U\in\Sigma_X$ and $V\in\Sigma_Y$; but in order to obtain a measurable space, we need to close this collection of basic measurable sets by the operations of Definition~\ref{defmessp}. See~\cite[Sec.~33]{HalmMeas} for more detail.

For stochastic operations $P(y_1|x_1)$ and $P(y_2|x_2)$, there also ought to exist a ``product'' operation corresponding to the idea of executing both $P(y_1|x_1)$ and $P(y_2|x_2)$ in parallel, analogous to the description of sequential execution described by~\eqref{stochcomp}. This product operation
\[
P(y_1|x_1) \, P(y_2|x_2) \: : \: (X_1\times X_2,\Sigma_{X_1}\otimes\Sigma_{X_2}) \longrightarrow (Y_1\times Y_2,\Sigma_{Y_1}\otimes\Sigma_{Y_2})
\]
is defined to map a pair of values $(x_1,x_2)$ to the unique measure on $Y_1\times Y_2$ having the property that it assigns to the set $U_1\times U_2\in \Sigma_{Y_1}\otimes\Sigma_{Y_2}$ the number $P(U_1|x_1)\cdot P(U_2|x_2)$; we refer to~\cite[Thm.~35.B]{HalmMeas} for the proof of uniqueness. We leave it to the reader to verify that this does indeed give a stochastic operation and end by noting that this product turns $\sop$ into a \emph{symmetric monoidal category}. In the particular case $X_1=X_2=I$, the product of stochastic operations specializes to a definition of product of measures.

\newpage
\section{Finite approximations of probability spaces}
\label{secfinapp}

Here we prove some technical results which are of crucial importance for the proof of Theorem~\ref{finitethm}. We refer to~\cite{HalmMeas} for the additional terminology and background from measure theory. For sets $A$ and $B$, we write
\[
A \bigtriangleup B \defin (A\setminus B) \cup (B\setminus A)
\]
for the symmetric difference.

\begin{lem}
\label{algebraapprox}
Let $(X,\Sigma,P)$ be a probability space, $\Gamma\subseteq\Sigma$ an algebra of sets and $f:X\to S$ a measurable function to some finite set $S$. Then for every $\delta > 0$, there exist sets $B_s\in\Gamma$ indexed by $s\in S$ such that $P\left(B_s \bigtriangleup f^{-1}(s)\right)<\delta$, and the $B_s$ partition $X$.
\end{lem}

The main point here is that $B_s\in\Gamma$: since $f$ is measurable, we already know that $f^{-1}(s)\in\Sigma$, and the problem is to approximate these $f^{-1}(s)$ by elements of $\Gamma$.

\begin{proof}
We use induction on $|S|$. If $|S|=1$, then there is nothing to prove, since we can simply take $B_s\defin X$. For general $S$, we choose any element $u\in S$. We then apply the approximation lemma of measure theory~\cite[13.D]{HalmMeas} to obtain a set $B_u\in\Gamma$ with $P\left( B_u\bigtriangleup f^{-1}(u)\right)<\delta/2$. 

The idea now is to restrict everything to $X{\setminus} B_u$ and apply the induction hypothesis to that. To do so, we define the function $f':X{\setminus} B_u\to S{\setminus}\{u\}$ by taking $f'(x) \defin f(x)$ for all $x$ with $f(x)\neq u$; we let $f'$ map the other $x$'s with $f(x)=u$ to an arbitrary element of $S{\setminus}\{u\}$. Then $f'$ differs from the restriction of $f$ to $X{\setminus}B_u$ on a set of measure less than $\delta/2$.

Applying the induction hypothesis to this situation, with $\delta$ replaced by $\delta/2$, gives a partition $X{\setminus} B_u = \bigcup_{s\in S{\setminus}\{u\}} B_s$ satisfying $P\left( B_s \bigtriangleup f'^{-1}(s)\right) < \delta/2$ for all $s\neq u$. Since we already know that $P\left(f'^{-1}(s) \bigtriangleup f^{-1}(s)\right) < \delta/2$ for all $s$, we conclude that
\begin{align*}
P\left( B_s \bigtriangleup f^{-1}(s)\right) & \leq P(B_s \bigtriangleup f'^{-1}(s)) + P(f'^{-1}(s) \bigtriangleup f^{-1}(s)) \\
 & < \delta/2 + \delta/2 = \delta,
\end{align*}
as desired.
\end{proof}

\begin{prop}
\label{finiteapprox}
Let $(X_k,\Sigma_k)_{k=1}^n$ be a family of measurable spaces equipped with a probability measure $P$ on the product $\prod_k X_k$ and $S$ a finite set. Then for every $\eps > 0$ and any measurable function
\[
f \: : \: X_1\times\ldots\times X_n \lra S
\]
there exist a finite set $S_k$ together with a measurable function $f_k : X_k\to S_k$ for each $k$, and a function
\[
g \: :\: S_1\times\ldots\times S_n \lra S
\]
such that the diagram
\[
\xymatrix{ X_1\times\ldots\times X_n \ar[rr]^f \ar[rd]_(.4){f_1\times\ldots\times f_n} && S \\
& S_1\times\ldots\times S_n \ar[ur]_(.6)g }
\]
commutes up to $\eps$, by which we mean that $f(x) = \left(g\circ (f_1\times\ldots\times f_n)\right)(x)$ holds with probability greater than $1-\eps$ with respect to $P$.
\end{prop}

The main point here is that the $S_k$ are finite: if we think of $f$ as a function which coarse-grains $X_1\times\ldots\times X_n$ down to a finite set $S$, then the claim is that this coarse-graining can be approximated arbitrarily well by first coarse-graining each $X_k$ down to a finite set $S_k$, and then coarse-graining $S_1\times\ldots\times S_n$ down to $S$.

\begin{proof}
The finite disjoint unions of products of measurable sets form an algebra of sets in $X_1\times\ldots\times X_n$~\cite[33.E]{HalmMeas} which generates the product $\sigma$-algebra $\Sigma_1\otimes\ldots\otimes\Sigma_n$. Hence Lemma~\ref{algebraapprox} gives a partition $X_1\times\ldots\times X_n = \bigcup_{s\in S} B_s$ such that $P\left( B_s \bigtriangleup f^{-1}(s) \right) < \delta$. We will determine the required value of $\delta$ later. By construction, we know that each $B_s$ is a finite union of sets of the form
\beq
\label{productA}
C_1^i\times\ldots\times C_n^i,
\eeq
where $C_k^i\in\Sigma_k$ and $i$ indexes all the required sets for all $s$ together. For fixed $k$, the finitely many sets $C_k^i$ generate a finite Boolean algebra $\Omega_k\subseteq\Sigma_k$ whose atoms we denote by $A_k^i$. Again by construction, the $B_s$'s are elements of $\Omega_1\otimes\ldots\otimes\Omega_n$, and hence they can also be written as finite unions of products of atoms,
\[
A_1^i\times\ldots\times A_n^i.
\]
We then define $S_k$ to be the set of all the atoms $A_k^i$, i.e.~the spectrum of the Boolean algebra $\Omega_k$. This $S_k$ is guaranteed to be finite since $\Omega_k$ is finite. There is a canonical map $f_k:X_k\to S_k$ which sends an $x\in X_k$ to the unique atom which contains it.

We need to define the map $g:S_1\times\ldots\times S_n\to S$. We do this by stipulating that a tuple of atoms $(A_1,\ldots,A_n)\in S_1\times\ldots\times S_n$ maps to the unique $s\in S$ for which $A_1\times\ldots\times A_n\subseteq B_s$; since the $B_s$'s partition $X$ and each of them is a finite union of products of atoms, that $s$ is indeed unique. 

Finally, we have to verify that this satisfies the required ``commutativity up to $\eps$'' for a suitable choice of $\delta$. This is easy to do once we realize that the composition $g\circ (f_1\times\ldots\times f_n)$ is precisely the map which sends all of $B_s$ to $s$. Hence the measure of the set of all tuples $(x_1,\ldots,x_n)$ on which $f$ differs from $g\circ(f_1\times\ldots\times f_n)$ is given by
\[
P\left( \bigcup_{s\in S} (B_s \bigtriangleup f^{-1}(s) )\right) \leq \sum_{s\in S} P\left( B_s\bigtriangleup f^{-1}(s) \right) < |S|\cdot\delta.
\]
So if we do the whole construction with $\delta \leq \eps/|S|$, the claim follows.
\end{proof}

\newpage

\bibliographystyle{unsrt}
\bibliography{bayesian_networks}

\begin{thebibliography}{10}

\bibitem{Bell}
John~S. Bell.
\newblock {On the Einstein-Podolsky-Rosen paradox}.
\newblock {\em Physics}, 1:195--200, 1964.

\bibitem{Bellreview}
Nicolas Brunner, Daniel Cavalcanti, Stefano Pironio, Valerio Scarani, and
  Stephanie Wehner.
\newblock Bell nonlocality, 2014.

\bibitem{Ekert}
Artur~K. Ekert.
\newblock Quantum cryptography based on {B}ell's theorem.
\newblock {\em Phys. Rev. Lett.}, 67:661--663, Aug 1991.

\bibitem{VVqkd}
Umesh Vazirani and Thomas Vidick.
\newblock Fully device independent quantum key distribution.
\newblock {\em Phys. Rev. Lett.}, 113:140501, 2014.

\bibitem{PAal}
S.~Pironio, A.~Ac{\'i}n, S.~Massar, A.~Boyer de~la Giroday, D.~N. Matsukevich,
  P.~Maunz, S.~Olmschenk, D.~Hayes, L.~Luo, T.~A. Manning, and C.~Monroe.
\newblock Random numbers certified by {B}ell’s theorem.
\newblock {\em Nature}, 464:1021, 2010.

\bibitem{VVrand}
Umesh Vazirani and Thomas Vidick.
\newblock Certifiable quantum dice: or, true random number generation secure
  against quantum adversaries.
\newblock In {\em Proceedings of the Forty-fourth Annual {ACM} Symposium on
  Theory of Computing}, STOC '12, pages 61--76, 2012.

\bibitem{Svet}
George Svetlichny.
\newblock Distinguishing three-body from two-body nonseparability by a
  {B}ell-type inequality.
\newblock {\em Phys.~Rev.~D}, 35:3066, 1987.

\bibitem{Pop}
Sandu Popescu.
\newblock Bell's inequalities and density matrices: Revealing ``hidden''
  nonlocality.
\newblock {\em Phys. Rev. Lett.}, 74:2619--2622, Apr 1995.

\bibitem{BGP}
Cyril Branciard, Nicolas Gisin, and Stefano Pironio.
\newblock Characterizing the nonlocal correlations created via entanglement
  swapping.
\newblock {\em Phys. Rev. Lett.}, 104:170401, 2010.

\bibitem{Fri}
Tobias Fritz.
\newblock Beyond {B}ell's theorem {I}: correlation scenarios.
\newblock {\em New J. Phys.}, 14:103001, 2012.

\bibitem{GWCAN}
Rodrigo Gallego, Lars~Erik W{\"u}rflinger, Rafael Chaves, Antonio Ac{\'i}n, and
  Miguel Navascu{\'e}s.
\newblock Nonlocality in sequential correlation scenarios.
\newblock {\em New J. Phys.}, 16:033037, 2014.

\bibitem{HardyOT}
Lucien Hardy.
\newblock The operator tensor formulation of quantum theory.
\newblock {\em Phil. Trans. R. Soc. A}, 28:3385--3417, 2012.

\bibitem{Pearl}
Judea Pearl.
\newblock {\em Causality}.
\newblock Cambridge University Press, Cambridge, second edition, 2009.
\newblock Models, reasoning, and inference.

\bibitem{EPR}
Albert Einstein, Boris Podolsky, and Nathan Rosen.
\newblock Can quantum-mechanical description of physical reality be considered
  complete?
\newblock {\em Phys. Rev.}, 47:777--780, 1935.

\bibitem{WS}
Christopher~J. Wood and Robert~W. Spekkens.
\newblock The lesson of causal discovery algorithms for quantum correlations:
  Causal explanations of {B}ell-inequality violations require fine-tuning.
\newblock {\em New. J. Phys.}, 17:033002, 2015.

\bibitem{HLP}
Joe Henson, Raymond Lal, and Matthew~F. Pusey.
\newblock Theory-independent limits on correlations from generalised bayesian
  networks.
\newblock {\em New J. Phys.}, 16:113043, 2014.

\bibitem{CLG}
Rafael Chaves, Lukas Luft, and David Gross.
\newblock Causal structures from entropic information: geometry and novel
  scenarios.
\newblock {\em New. J. Phys.}, 16:043001, 2014.

\bibitem{CLMGJS}
Rafael Chaves, Lukas Luft, Thiago~O. Maciel, David Gross, Dominik Janzing, and
  Bernhard Sch\"olkopf.
\newblock Inferring latent structures via information inequalities.
\newblock In {\em Proceedings of the 30th Conference on Uncertainty in
  Artificial Intelligence}, pages 112--121. AUAI Press, 2014.

\bibitem{CMG}
Rafael Chaves, Christian Majenz, and David Gross.
\newblock Information-theoretic implications of quantum causal structures.
\newblock {\em Nat. Comm.}, 6:5766, 2015.

\bibitem{CKBG}
Rafael Chaves, Richard Kueng, Jonatan Bohr~Brask, and David Gross.
\newblock A unifying framework for relaxations of the causal assumptions in
  bell's theorem.
\newblock {\em Phys. Rev. Lett.}, 114:140403, 2015.

\bibitem{cc}
Frank Arntzenius.
\newblock Reichenbach's common cause principle.
\newblock In Edward~N. Zalta, editor, {\em The Stanford Encyclopedia of
  Philosophy}. Stanford University, 2010.
\newblock
  \href{http://plato.stanford.edu/archives/fall2010/entries/physics-Rpcc/}{plato.stanford.edu/archives/fall2010/entries/physics-Rpcc/}.

\bibitem{OCB}
Ognyan Oreshkov, Fabio Costa, and {\v{C}}aslav Brukner.
\newblock Quantum correlations with no causal order.
\newblock {\em Nat. Comm.}, 3:1092, 2012.

\bibitem{Bierhorst}
Peter Bierhorst.
\newblock A rigorous analysis of the {C}lauser-{H}orne-{S}himony-{H}olt
  inequality experiment when trials need not be independent.
\newblock {\em Found. Phys.}, 44:736--761, 2014.

\bibitem{Henson}
Joe Henson.
\newblock The causal set approach to quantum gravity, 2006.
\newblock \href{http://arxiv.org/abs/gr-qc/0601121}{arXiv:gr-qc/0601121}.

\bibitem{BLMPPR}
Jonathan Barrett, Noah Linden, Serge Massar, Stefano Pironio, Sandu Popescu,
  and David Roberts.
\newblock Nonlocal correlations as an information-theoretic resource.
\newblock {\em Phys. Rev. A}, 71:022101, Feb 2005.

\bibitem{LG}
Anthony~J. Leggett and Anupam Garg.
\newblock Quantum mechanics versus macroscopic realism: is the flux there when
  nobody looks?
\newblock {\em Phys. Rev. Lett.}, 54:857, 1985.

\bibitem{Fine}
Arthur Fine.
\newblock Hidden variables, joint probability, and the {B}ell inequalities.
\newblock {\em Phys. Rev. Lett.}, 48(5):291--295, Feb 1982.

\bibitem{Upper}
Daniel~Ray Upper.
\newblock {\em Theory and Algorithms for Hidden {M}arkov Models and Generalized
  Hidden {M}arkov Models}.
\newblock PhD thesis, University of California at Berkeley, 1989.
\newblock Available at
  \href{http://csc.ucdavis.edu/~cmg/papers/TAHMMGHMM.pdf.gz}{csc.ucdavis.edu/$\sim$cmg/papers/TAHMMGHMM.pdf.gz}.

\bibitem{CDP}
Giulio Chiribella, G.~Mauro D'Ariano, and Paolo Perinotti.
\newblock Probabilistic theories with purification.
\newblock {\em Phys. Rev. A}, 81:062348, 2010.

\bibitem{CP}
Bob Coecke and Eric~Oliver Paquette.
\newblock Categories for the practising physicist.
\newblock In {\em New {S}tructures for {P}hysics}, pages 173--286. Springer,
  2011.

\bibitem{Spivak}
David Spivak.
\newblock Category theory for scientists (old version), 2013.
\newblock \href{http://arxiv.org/abs/1302.6946}{arXiv:1302.6946}.

\bibitem{CLcausal}
Bob Coecke and Raymond Lal.
\newblock Causal categories: relativistically interacting processes.
\newblock {\em Found. Phys.}, 43(4):458--501, 2012.

\bibitem{Coecke}
Bob Coecke.
\newblock Quantum picturalism.
\newblock {\em Contemporary Physics}, 51(1):59--83, 2010.

\bibitem{JS}
Andr{\'e} Joyal and Ross Street.
\newblock The geometry of tensor calculus, {I}.
\newblock {\em Adv.~Math.}, 88:55--112, 1991.

\bibitem{Leinster}
Tom Leinster.
\newblock Higher operads, higher categories, 2003.
\newblock \href{http://arxiv.org/abs/math/0305049}{arXiv:math/0305049}.

\bibitem{Cns}
Bob Coecke.
\newblock Terminality implies non-signalling.
\newblock In {\em Proceedings 11th workshop on Quantum Physics and Logic},
  volume 172, pages 27--35, 2014.

\bibitem{SSR}
David~I. Spivak, Patrick Schultz, and Dylan Rupel.
\newblock String diagrams for traced and compact categories are oriented
  1-cobordisms, 2015.
\newblock \href{http://arxiv.org/abs/1508.01069}{arXiv:1508.01069}.

\bibitem{Barrett}
Jonathan Barrett.
\newblock Information processing in generalized probabilistic theories.
\newblock {\em Phys. Rev. A}, 75(3):032304, Mar 2007.

\bibitem{SB}
Anthony Short and Jonathan Barrett.
\newblock Strong nonlocality: A trade-off between states and measurements.
\newblock {\em New J. Phys.}, 12:033034, 2010.

\bibitem{Gisin}
Nicolas Gisin.
\newblock Non-realism: deep thought or soft option?
\newblock {\em Found. Phys.}, 42:80--85, 2012.

\bibitem{Gisin2}
Nicolas Gisin.
\newblock A possible definition of a \textit{realistic} physics theory, 2014.
\newblock \href{http://arxiv.org/abs/1401.0419}{arXiv:1401.0419}.

\bibitem{HalmHilb}
Paul~R. Halmos.
\newblock {\em A {H}ilbert Space Problem Book}, volume~19 of {\em Graduate
  Texts in Mathematics}.
\newblock Springer, second edition, 1982.

\bibitem{KR}
Richard~V. Kadison and John~R. Ringrose.
\newblock {\em Fundamentals of the theory of operator algebras. {V}ol. {I}},
  volume 100 of {\em Pure and Applied Mathematics}.
\newblock Academic Press Inc. Harcourt Brace Jovanovich Publishers, 1983.
\newblock Elementary theory.

\bibitem{Conway}
John~B. Conway.
\newblock {\em A Course in Operator Theory}, volume~21 of {\em Graduate Studies
  in Mathematics}.
\newblock American Mathematical Society, 1999.

\bibitem{NC}
Michael Nielsen and Isaac Chuang.
\newblock {\em Quantum Computation and Quantum Information}.
\newblock Cambridge University Press, 2000.

\bibitem{GY}
Mark Gales and Steve Young.
\newblock The application of hidden {M}arkov models in speech recognition.
\newblock {\em Foundations and Trends in Signal Processing}, 1(3):195--304,
  2008.

\bibitem{Eddy}
Sean~R. Eddy.
\newblock What is a hidden markov model?
\newblock {\em Nature Biotechnology}, 22:1315--1316, 2004.

\bibitem{Ghahr}
Zoubin Ghahramani.
\newblock An introduction to hidden {M}arkov models and {B}ayesian networks.
\newblock {\em Int. J. Patt. Recogn. Artif. Intell.}, 15(1):9--42, 2001.

\bibitem{LS}
Matthew Leifer and Robert Spekkens.
\newblock Towards a formulation of quantum theory as a causally neutral theory
  of {B}ayesian inference.
\newblock {\em Phys. Rev. A}, 88:052130, 2013.

\bibitem{HalmMeas}
Paul~R. Halmos.
\newblock {\em Measure theory}.
\newblock D. Van Nostrand Company, Inc., New York, N. Y., 1950.

\bibitem{PP}
Prakash Panangaden.
\newblock {\em Labelled Markov Processes}.
\newblock Imperial College Press, 2009.

\bibitem{BF}
John~C. Baez and Tobias Fritz.
\newblock A {B}ayesian characterization of relative entropy.
\newblock {\em Theory and Applications of Categories}, 29:421--456, 2014.

\end{thebibliography}

\end{document}